\documentclass[final]{amsart}

\usepackage{amsthm,amsmath,amsfonts,amssymb}
\usepackage{graphicx}
\usepackage[comma, authoryear]{natbib}
\usepackage[colorlinks,citecolor=blue,urlcolor=blue]{hyperref}
\usepackage[utf8]{inputenc}
\usepackage{multirow}
\usepackage{mathtools}
\usepackage{ulem}
\usepackage{setspace}
\usepackage{pdflscape}
\usepackage[color,notref,notcite]{showkeys}
\usepackage{enumerate}
\definecolor{labelkey}{rgb}{0.6,0,1}

\theoremstyle{plain}
\newtheorem{thm}{Theorem}
\newtheorem{prop}{Proposition}
\newtheorem{lem}{Lemma}
\newtheorem{corollary}{Corollary}
\newtheorem{ass}{Assumption}

\theoremstyle{definition}

\newtheorem{remark}{Remark}
\newtheorem{example}{Example}

\def\disp{\displaystyle}

\newcommand{\dC}{\ensuremath{\mathbb{C}}}

\newcommand{\dG}{\ensuremath{\mathbb{G}}}

\newcommand{\I}{\ensuremath{\mathbb{I}}}

\newcommand{\dM}{\ensuremath{\mathbb{M}}}
\newcommand{\dN}{\ensuremath{\mathbb{N}}}

\newcommand{\dR}{\ensuremath{\mathbb{R}}}

\newcommand{\dZ}{\ensuremath{\mathbb{Z}}}

\newcommand{\cA}{\ensuremath{\mathcal{A}}}
\newcommand{\cB}{\ensuremath{\mathcal{B}}}
\newcommand{\cC}{\ensuremath{\mathcal{C}}}

\newcommand{\cG}{\ensuremath{\mathcal{G}}}

\newcommand{\cI}{\ensuremath{\mathcal{I}}}

\newcommand{\cK}{\ensuremath{\mathcal{K}}}
\newcommand{\cL}{\ensuremath{\mathcal{L}}}
\newcommand{\cM}{\ensuremath{\mathcal{M}}}
\newcommand{\cN}{\ensuremath{\mathcal{N}}}
\newcommand{\cO}{\ensuremath{\mathcal{O}}}

\newcommand{\cR}{\ensuremath{\mathcal{R}}}
\newcommand{\cS}{\ensuremath{\mathcal{S}}}

\def\i1{ [-\infty,\infty]}

\def\cor{ {\rm cor~}}
\def\cov{ {\rm cov~}}
\def\var{ {\rm var~}}
\def\card{ {\rm card~}}
\def\sn{n^{1/2} }
\def\sni{n^{-1/2} }

\def\bs{\mathbf{s}}

\def\bu{\mathbf{u}}
\def\bv{\mathbf{v}}
\def\bx{\mathbf{x}}

\def\bF{\mathbf{F}}
\def\bG{\mathbf{G}}

\def\bK{\mathbf{K}}

\def\bR{\mathbf{R}}

\def\bU{\mathbf{U}}
\def\bV{\mathbf{V}}
\def\bX{\mathbf{X}}

\def\bdelta{{\boldsymbol\delta}}

\def\btheta{{\boldsymbol\theta}}

\def\setd{\{1,\ldots,d\}}
\def\setn{\{1,\ldots,n\}}

\def\setp{\{1,\ldots,p\}}

\def\unif{{\rm U}(0,1)}

\newcommand{\BB}[1]{\textcolor{black}{#1}}

\begin{document}

\title[Tests of independence and randomness for arbitrary data]{Tests of independence and randomness for arbitrary data using copula-based covariances}

\author{Bouchra R. Nasri}
\address{Département de médecine sociale et préventive, École de santé publique, Université de Montréal, C.P. 6128, succursale Centre-ville
Montréal (Québec)  H3C 3J7 }
\email{bouchra.nasri@umontreal.ca}

\author{Bruno N. R\'{e}millard}
\address{GERAD and Department of Decision Sciences, HEC Montr\'{e}al\\
3000, che\-min de la C\^{o}\-te-Sain\-te-Ca\-the\-ri\-ne,
Montr\'{e}al (Qu\'{e}\-bec), Canada H3T 2A7}
\email{bruno.remillard@hec.ca}

\thanks{Funding in partial support of this work was provided by the Fonds
qu\'e\-b\'e\-cois de la re\-cher\-che en sant\'e and the Natural Sciences and Engineering Research Council of Canada.}


\begin{abstract}
In this article, we study tests of independence for data with arbitrary distributions in the non-serial case, i.e., for independent and identically distributed random vectors, as well as in the serial case, i.e., for time series.  These tests are derived from copula-based  covariances and their multivariate extensions using M\"obius transforms.
 We find the asymptotic distributions of these statistics under the null hypothesis of independence or randomness, as well as under contiguous alternatives. This enables us to find out locally most powerful test statistics for some alternatives, whatever the margins.  Numerical experiments are performed for Wald's type combinations of these statistics to assess the finite sample performance.
\end{abstract}




\maketitle

\section{Introduction}
\label{sec:intro}

In many cases, \BB{tests of independence} using simple dependence measures  like Kendall's tau and Spearman's rho, perform as well as tests based on more complex statistics  constructed from empirical processes \citep{Blum/Kiefer/Rosenblatt:1961}, and are generally much faster to compute. However, tests based of such measures are not always consistent.
Nevertheless, tests of independence or randomness based on copulas should always be performed. Here, we are interested in copula-based   of dependence measures for two cases: (i)  the non-serial case, where we have a sample of iid observations  $\bX_i = (X_{i1},\ldots,X_{id})\sim H$, with vector of margins $\bF=(F_1,\ldots, F_d)$, $d\ge 2$, and we want to test the null hypothesis of independence, i.e., $H_0: H(x_1,\ldots,x_d) = \prod_{j=1}^d F_j(x_j)$, for all $\bx=(x_1,\ldots, x_d)\in \dR^d$; (ii) the serial case, i.e., we have a stationary time series $(Y_t)_{t\ge 1}$ with common margin $F$, and we want to test the null hypothesis of randomness, i.e., $H_0: P(Y_1\le y_1, \ldots, Y_n\le y_n) = \prod_{t=1}^n F(y_t)$ for all $(y_1,\ldots, y_n)\in \dR^n$, $n\ge 2$. In the serial case, one defines the random vectors $\bX_t = (Y_t,\ldots, Y_{t+1-d})$, and hereafter,  the series $Y$ is extended in a circular way by setting $Y_{t+n}=Y_t$ for all $t\in\dZ$.

In the bivariate case, when the margins are continuous, many copula-based  dependence measures  are theoretically defined as the correlation $\varrho_{\bK}(C) = \cor\left\{K_1^{-1}(U_{i1}), K_2^{-1}(U_{i2})\right\}$,  since by continuity of the margins, $\bU_i = (U_{i1},U_{i2}) = \bF(\bX_i) \sim C$, for a unique  copula $C$. Here, $\bK=(K_1,K_2)$ is a given vector of cdfs, with mean $\mu_1,\mu_2$,  and \BB{finite} variances $\sigma_1^2,\sigma_2^2$.  These requirements on $K_1,K_2$ are necessary in order for the correlation to exist. The value under independence is clearly $0$. Next, by Hoeffding's identity \citep{Hoeffding:1940},
\begin{equation}
    \label{eq:rhocov}
    \varrho_{\bK}(C) = \frac{\gamma_\bK(C)}{\sigma_1\sigma_2} = \frac{1}{\sigma_1\sigma_2} \int_{\dR^2}\left[ C\left\{K_1(x_1),K_2(x_2)\right\}-K_1(x_1)K_2(x_2)\right]dx_1 dx_2,
\end{equation}
since $\left( K_1^{-1}(U_1), K_2^{-1}(U_2)\right) $ has joint cdf $C\circ \bK$.
For example, suppose that $(U_1,U_2)\sim C$. Then,  taking $K_1=K_2=D$, where $D$ is the cdf of the uniform distribution over $(0,1)$, one obtains  Spearman's correlation $\rho_S(C) = 12E(U_1U_2)-3$. The case $K_1=K_2 = \Phi$, where $\Phi$ is the cdf of the standard Gaussian distribution, yields the van der Waerden coefficient $\rho_{vdw}(C)  =E\left\{\Phi^{-1}(U_1)\Phi^{-1}(U_2)\right\}$. Savage's coefficient corresponds to $K_j(x) \equiv  1-e^{-x}$, $\mu_j\equiv 1$, with the convention that $0\log{0}=0$, while if $K_1=K_2$ is the cdf of a Bernoulli(1/2), one gets Blomqvist's coefficient $ 4C(1/2,1/2)-1$.
Note that by definition, when $K_1=K_2$,
$\varrho_K(C_+)=1$ for the complete dependence, where $C_+(u_1,u_2)=\min(u_1,u_2)$ is the Fr\'echet-Hoeffding upper bound.
However, when $K_1\neq K_2$, the covariance $\gamma_{K_1,K_2}(C)$ must be divided by $E\left\{ K_1^{-1}(U_1)K_2^{-1}(U_1)\right\}  -\mu_1\mu_2$ to give $1$ for complete dependence.  Blest's coefficient \citep{Blest:2000,Genest/Plante:2003} can be seen as such an example if one considers a natural modification. In fact, Blest's coefficient  has been originally defined as the covariance  between $(1-U_1)^2$ and $U_2$. An obvious modification is obtained by taking $K_1(u)=u^{1/2}$, $u\in [0,1]$, $K_2=D$, so the modified coefficient is $12E(U_1^2 U_2)-2$, normalised to give $1$ for complete dependence. \cite{Genest/Plante:2003} also proposed a symmetrised Blest's coefficient, which, in our general setting, amounts to defining
 $$
 \gamma_{K_1,K_2}^*(C) = \dfrac{\gamma_{K_1,K_2}(C)+ \gamma_{K_2,K_1}(C)}{2}.
 $$
 Not all copula-based dependence measures  are defined by a covariance (up to a constant), a well-known example being  Kendall's tau, defined by
 $\disp \tau(C) = 4\int_{(0,1)^2}  C(u_1,u_2) dC(u_1,u_2)-1 = 4E\{C(U_1,u_2)\}-1$.
 However, Kendall's tau and Spearman's rho have an equivalent limiting distribution, even under a sequence of contiguous alternatives. There are also other interesting copula-based  dependence measures, namely $\phi$-dependence measures, recently studied in \cite{Geenens/LafayedeMicheaux:2022}, and defined by $\disp \int_0^1\int_0^1 \phi\{c(u_1,u_2)\}du_2du_2$, where $\phi$ is a convex function with $\phi(1)=0$. The case $\phi(t) = t\log{t}$ yields mutual information, while $\phi(t) = \left(t^{1/2}-1\right)^2$ yields Hellinger's correlation.

Estimating the dependence measures defined previously is relatively straightforward. In the continuous case,  the copula is replaced by the empirical copula
$$
\widehat C_n(u_1,u_2) = n^{-1}\sum_{i=1} \I\left\{\frac{n}{n+1}F_{n1}(X_{i1})\le u\right\}\I\left\{\frac{n}{n+1}F_{n2}(X_{i2})\le u_2\right\},
\quad u_1,u_2\in [0,1],
$$
where $\disp F_{nj}(x) = n^{-1}\sum_{i=1}^n \I\{X_{ij}\le x\}$, $x\in\dR$, $j\in \{1,2\}$.
In fact,  $\gamma_{\bK}(C)$ can be estimated by $\gamma_{\bK}\left(\widehat C_n\right)$, and according to \cite{Genest/Remillard:2004}, one has
$$
\gamma_{\bK}\left(\widehat C_n\right) = \int_{\dR^2} \left[ \widehat C_n\left\{K_1(x_1),K_2(x_2)\right\}-K_1(x_1)K_2(x_2)\right]dx_1 dx_2.
$$
Asymptotic limits and their representations are easier to work with the latter representation, being a linear functional of the empirical process $\widehat \dC_n(u_1,u_2) = \sn \left\{\hat C_n(u_1,u_2) - u_1 u_2\right\}$.
Tests of independence based on these dependence measures work well when the margins are continuous. However, for applications, there is a need to test independence in the more general setting of arbitrary distributions, i.e., when at least of the margins is not continuous. In this case, since there are ties, one might be tempted to replace the ranks by the mid-ranks. However, the asymptotic distribution might not be simple enough and it makes sense to test independence using copula-based extensions of these dependence measures.

The main problem here is that the copula is not unique. If $\bX\sim H$, there are infinitely many copulas satisfying Sklar's equation
$H = C\circ \bF$ \citep{Sklar:1959}. To construct solutions for this equation, for any copula $\cC$, take $\bV\sim \cC$ independent of $\bX\sim H$ and set
$\bU = \psi_{\bF}(\bX,\bV) $, where $U_j = \psi_{F_j}(X_j,V_j) = F_j(X_j-)+V_j \Delta_{F_j}(X_i)$, with
$F_j(x-)=P(X_j < x)$ and  $\Delta_{F_j}(x) = F_j(x)-F_j(x-)=P(X_j=x)$, $j\in \setd$. It is known \citep{Ferguson:1967,Ruschendorf:1981, Neslehova:2007, Brockwell:2007}
that for any $j\in\setd$,  $U_j\sim \unif$, and the joint cdf  $C_\cC$ of $\bU$ is a copula satisfying Sklar's equation.
In addition, there is one interesting copula $C^\maltese$ in this family, the so-called multilinear copula, obtained by taking $\cC=\Pi$, the independence copula, i.e., $\Pi(\bu)=\prod_{j=1}^d D(u_j)$. One interesting property of $C^\maltese$ is that
if  $H(\bx)=\prod_{j=1}^d F_j(x_j)$, then $C_\cC=\Pi$ if and only if $C_\cC=C^\maltese$.
As a by-product, taking the empirical joint cdf $H_n$ with the vector of margins $\bF_n=(F_{n_1},F_{n2})$, one obtains the empirical multilinear copula
$\widehat C_n^\maltese$, for which an explicit expression will be given in the next section.
Note that contrary to $\widehat C_n$, $\widehat C_n^\maltese$ is a genuine copula, so all dependence measures presented before can
be computed with $C^\maltese$ and its empirical counterpart $\widehat C_n^\maltese$. This is the approach  we propose here. \BB{However, since
these statistics are not margin-free anymore, we will no longer call them dependence measures. See, e.g., \cite[Section 2]{Geenens/LafayedeMicheaux:2022} for an interesting discussion on  dependence.} Note that the asymptotic behaviour of the associated versions of Kendall's tau and Spearman's rho has been studied in \cite{Genest/Neslehova/Remillard:2014}, and tests of independence based on $\widehat C_n^\maltese$ were proposed in \cite{Genest/Neslehova/Remillard/Murphy:2019}, while
in the serial case, tests of randomness based on the serial version $\widehat C_n^{\maltese,s}$ have  been studied in \cite{Nasri:2022}, as well as the asymptotic behaviour of the serial versions of Kendall's tau and Spearman's rho. \BB{The results in \citet{Nasri:2022} extend those of \cite{Kojadinovic/Yan:2011b} obtained for continuous observations.}

The main aim of this article is to define tests of independence and tests of randomness using bivariate and multivariate extensions of the copula-based  dependence measures when the margins are arbitrary, to find explicit expressions of these statistics, and to study their asymptotic behaviour. We will also look at the asymptotic distribution of the test statistics under a sequence of contiguous alternatives to be able to suggest locally powerful tests for given dependence models, in the same spirit as \cite{Genest/Verret:2005} did in the bivariate case for continuous margins. To this end, we also present a new representation of the multilinear copulas in the serial and non-serial cases that enables us to perform calculations more easily.
Note that in both \cite{Genest/Neslehova/Remillard/Murphy:2019} and \cite{Nasri:2022}, the main focus was on using Cram\'er-von Mises statistics of related multilinear processes, which is not done here.

In Section \ref{sec:processes}, we recall the definitions and properties of multilinear copulas in a  serial setting \citep{Nasri:2022} and non-serial setting \cite{Genest/Neslehova/Remillard/Murphy:2019}, together with their associated M\"obius transforms. Next, in Section \ref{sec:depmeasures}, we define the serial and non-serial versions of the proposed statistics extending the dependence measures, providing explicit formulas that are easy to implement, and we study their asymptotic behaviour under the null hypothesis of independence or randomness. Multivariate extensions similar to those defined in \cite{Genest/Remillard:2004} and \cite{Genest/Neslehova/Remillard:2014}
will also be studied. In addition, under additional moment conditions, one finds the asymptotic limits of these statistics when the null hypothesis is not satisfied. Next, in Section \ref{sec:contiguity}, we  study the asymptotic behaviour of the proposed test statistics under a sequence of contiguous alternatives, using the results of \cite{Genest/Neslehova/Remillard/Murphy:2019} and \cite{Nasri:2022}. This  enables us to find the locally most powerful tests amongst the class of the proposed tests statistics. We  also discuss how to combine the proposed tests statistics. Finally, numerical experiments are performed in Section \ref{sec:num} to assess the power of the tests for finite samples.

\section{Multilinear copulas and associated empirical processes}\label{sec:processes}

From now on, we consider the following two settings: the non-serial case and the serial case. In the non-serial case, we have independent and identically distributed (iid) random vectors $\bU_1,\ldots, \bU_n\sim C$, for a given copula $C$, and the observations are $\bX_i=\bF^{-1}(\bU_i)$, $i\in \setn$. In the serial setting,
 we have a stationary and ergodic sequence  of random variables $(U_t)_{t\ge 1}$, and the observed time series is $Y_t = F^{-1}(U_t)$, $t\in \setn$.
  We can now define the multilinear copula. For any $j\in\setd$, set $J_{F_j}(x_j,u_j) = E\left[\psi_{F_j}(X_j,u_j)|X_j=x_j\right] = P\left\{ F_j(x_j-)+ V_j \Delta_{F_j}(x_j)\le u_j \right\}$. Then,
$\disp
J_{F_j}(x_j,u_j) = \left\{ \begin{array}{cc}
\I\{F_j(x_j)\le u_j\}, & \text{ if }\Delta_{F_j}(x_j)=0,\\
D\left\{\frac{u_j-F_j(x_j-)}{\Delta_{F_j}(x_j)}\right\}, & \text{ if }\Delta_{F_j}(x_j)>0.
\end{array}\right.$, where $D$ is the cdf of $U\sim \unif$. Note that when $\Delta_{F_j}(x_j)>0$, $J_{F_j}(x_j,u_j)=0$ if $u_j \le F_j(x_j-)$, $J_{F_j}(x_j,u_j)=1$ if $u_j \ge F_j(x_j)$, and $J_{F_j}(x_j,u_j)=\dfrac{u_j-F_j(x_j-)}{\Delta_{F_j}(x_j)}$ if $F_j(x_j-)\le u_j \le F_j(x_j)$.
Using properties of conditional expectations, one obtains
\begin{equation}
    \label{eq:Cmaltese}
    C^\maltese(\bu) = E\left\{\prod_{j=1}^d J_{F_j}(X_j,u_j)\right\}, \quad \bu \in [0,1]^d.
\end{equation}
As a result,
\begin{equation}
    \label{eq:Cnmaltese}
    \widehat C_n^\maltese(\bu) = n^{-1}\sum_{i=1}^n \prod_{j=1}^d J_{F_{nj}}(X_{ij},u_j) =
    n^{-1}\sum_{i=1}^n \prod_{j=1}^d D\left\{\frac{u_j-F_{nj}(X_{ij}-)}{\Delta_{F_{nj}}(X_{ij})}\right\}
    , \quad \bu \in [0,1]^d.
\end{equation}
This new expression is different  from what appears in the literature, e.g., \cite{Genest/Neslehova/Remillard:2017,Genest/Neslehova/Remillard/Murphy:2019}, but it is easier to manipulate for our purposes. In fact, $\hat C_n^\maltese$ was previously defined by
$$
\hat C_n^\maltese(\bu)  = n^{-1} \sum_{i=1}^n \prod_{j=1}^d  \Bigl[  \lambda_{F_{nj}}(u_j) \I \{X_{ij}\le
F_{nj}^{-1}(u_j) \} \\ + \{1-\lambda_{F_{nj}}(u_j)\}\I \{X_{ij} < F_{nj}^{-1}(u_j) \} \Bigr],
$$
where, for any cdf $G$ and $u\in (0,1)$,
$\lambda_G(u) = \left\{
\begin{array}{cc}
\dfrac{ u - G \left\{G^{-1} (u)-\right\}} {\Delta_G\left\{G^{-1}(u)\right\}}, & \quad \Delta_{G} \left\{G^{-1}(u)\right\} > 0, \\
1, & \quad \text{otherwise}. \end{array}\right.
$.
Next, the empirical serial multilinear copula, first defined and studied in \cite{Nasri:2022}, can also be written as
\begin{equation}
    \label{eq:copempserial}
\widehat C_n^{\maltese,s}(\bu) = n^{-1}\sum_{t=1}^n  \prod_{j=1}^d D\left\{\frac{u_j-F_{n}(Y_{t+1-j}-)}{\Delta_{F_{n}}(Y_{t+1-j})}\right\},  \qquad \bu \in [0,1]^d,
\end{equation}
where $\disp F_n(y) = n^{-1}\sum_{t=1}^n \I\{Y_{t+1-j}\le y\}$, $y\in \dR$, for any $j\in \setd$, using the circular construction.
Further define the empirical multilinear processes $\widehat\dC_n^\maltese = \sn\left(\widehat C_n^\maltese-\Pi\right)$ and
$\widehat\dC_n^{\maltese,s} = \sn\left(\widehat C_n^{\maltese,s}-\Pi\right)$.
 Next,
 let $\cN_d$ be the set of all subsets $A$ of $\setd$ with $\card(A) = |A|>1$, and let
$\cS_d$ be the set of all elements $A$ of $\cN_d$ with $A\ni 1$.  It has been shown, e.g., \cite{Genest/Remillard:2004, Ghoudi/Remillard:2018, Genest/Neslehova/Remillard/Murphy:2019, Nasri:2022}, that M\"obius transforms of empirical processes have nice asymptotic properties for tests of independence or tests of randomness. To this end, define
\begin{equation}
    \label{eq:GnAnonserial}
    \dG_{A,n}^{\maltese}(\bu) =  \cM_A\left(\widehat \dC_n^{\maltese}\right)(\bu)=
    n^{-1/2}\sum_{i=1}^n  \prod_{j\in A} \left[
  D\left\{\frac{u_j-F_{nj}(X_{ij}-)}{\Delta_{F_{nj}}(X_{ij})}\right\}-u_j\right], \quad A\in\cN_d,
\end{equation}
\begin{equation}
    \label{eq:GnAserial}
    \dG_{A,n}^{\maltese,s}(\bu) =  \cM_A\left(\widehat \dC_n^{\maltese,s}\right)(\bu)= n^{-1/2}\sum_{t=1}^n  \prod_{j\in A} \left[
  D\left\{\frac{u_j-F_{n}(Y_{t+1-j}-)}{\Delta_{F_{n}}(Y_{t+1-j})}\right\}-u_j\right], \quad A\in\cS_d,
\end{equation}
where the M\"obius transform $\cM_A$ is defined in Appendix \ref{app:aux}.
Next,  for any $s,t\in [0,1]$,  and any cdf $G$, set
  \begin{equation}\label{eq:gammaG}
\Gamma_{G}(s,t) = s\wedge t -st - \sum_{x: \Delta_G(x)>0}\I\{G(x-)\le  s\wedge t \le s\vee t \le G(x)\}
 \frac{ \left\{(s\wedge t)-G(x-)\right\}\left\{G(x)- s\vee s)\right\}}{\Delta_{G}(x)}.
\end{equation}
The main findings of \cite{Genest/Neslehova/Remillard/Murphy:2019} and \cite{Nasri:2022} that we need can be summarised as follows:
\begin{thm}\label{thm:main}
Under the null hypothesis of independence, $\left\{\dG_{A,n}^\maltese: A\in \cN_d\right\}$ converge jointly in $\ell^\infty\left((0,1)^d\right)$ to independent centred Gaussian processes $\left\{\dG_{A}^\maltese: A\in \cN_d\right\}$, where
$\disp E\left\{\dG_A^\maltese(\bu) \dG_A^\maltese(\bv) \right\} =
\prod_{j\in A} \Gamma_{F_j}(u_j,v_j)$. 
Under the null hypothesis of randomness,
 $\left\{\dG_{A,n}^{\maltese,s}: A\in \cS_d\right\}$ converge jointly in $\ell^\infty\left((0,1)^d\right)$ to
 independent centred Gaussian processes  $\left\{\dG_{A}^{\maltese,s}: A\in \cS_d\right\}$,
where
$\disp E\left\{\dG_A^{\maltese,s}(\bu) \dG_A^{\maltese,s}(\bv) \right\} =
\prod_{j\in A} \Gamma_F(u_j,v_j)$.
\end{thm}
\begin{remark}\label{rem:bivcop}
The formulas for the covariances in Theorem \ref{thm:main} follows from (D.6) and \eqref{eq:gammaG} in \cite{Nasri:2022}. One can check that for any $s,t\in [0,1]$, $\Gamma_G(s,t)\ge 0$ with equality if and only if $s\wedge t =0$ or $s\vee t=1$. It is interesting to note that for sets $A$ of size $2$, $\dG_{A,n}^\maltese$ and $\dG_{A,n}^{\maltese,s}$ are empirical multilinear copula processes. In fact,  for $A=\{j,k\}\in \cN_d$, $j<k$,
$\dG_{A,n}^{\maltese}(u_1,u_2) = \sn \left\{ \widehat C_{A,n}^{\maltese}(u_1,u_2)-u_1u_2\right\}$,
where $\widehat C_{A,n}^{\maltese}$ is the  empirical multilinear copula for the pairs $(X_{ij},X_{ik})$, $i\in\setn$.
Similarly, for any $A = \{1,1+\ell\}\in \cS_d$,
$\dG_{A,n}^{\maltese,s}(u_1,u_{2}) = \sn \left\{\widehat C_{A,n}^{\maltese,s}(u_1,u_{2})  -u_1u_2\right\}$,
where $C_{A,n}^{\maltese,s}$ is the empirical multilinear copula for the pairs $(Y_t,Y_{t-\ell})$, $t\in\setn$.
\end{remark}

\section{Statistics for testing independence or randomness for arbitrary distributions}\label{sec:depmeasures}

From now on, let $\bK = (K_1,\ldots, K_d)$ be a vector of margins with mean $\mu_j$ and \BB{finite} variance $\sigma_j^2$,  $j\in\setd$, and define
 $\disp \cL_{K_j}(u) = \int_0^u K_j^{-1}(v)dv$.
Next, for any $j\in \setd$, and any cdf $G$, define
$\disp
\cK_{j,G}(x) = \int_0^1 K_j^{-1}\left\{G(x-)+s \Delta_G(x)\right\}ds$.
Then $\cK_{j,G}(x)  = K_j^{-1}\{ G(x)\} $, if $G$ is continuous at $x$,  and
    $\cK_{j,G}(x)  = \dfrac{\cL_{K_j}\{G(x)\} -\cL_{K_j}\{G(x-)\} }{\Delta_G(x)}$,
if $G$ is not continuous at $x$. The extension of the covariance measures is defined in the following way:\\
In the non-serial case, for any $A\in\cN_d$, set
$\disp
\gamma_{\bK,A}\left(\widehat C_n^\maltese\right)  = \sni (-1)^{|A|} \int_{\dR^A}  \dG_{A,n}^\maltese\left\{\bK(\bx)\right\}d\bx$,
while in the serial case, set
$\disp
\gamma_{\bK,A}\left(\widehat C_n^{\maltese,s}\right)=  \sni (-1)^{|A|} \int_{\dR^A}  \dG_{A,n}^{\maltese,s}\left\{\bK(\bx)\right\}d\bx$, $A\in\cS_d$.
It then follows from Proposition \ref{prop:DG} in Appendix \ref{app:aux} that for any $A\in \cN_d$, in the non-serial case,
\begin{equation}\label{eq:TnAnonserial}
    \gamma_{\bK,A}\left(\widehat C_n^\maltese\right)  = n^{-1}\sum_{i=1}^n \prod_{j\in A}\left\{\cK_{j,F_{nj}}(X_{ij})-\mu_j\right\},
\end{equation}
while in the serial case, for any $A\in \cS_d$,
\begin{equation}\label{eq:TnAserial}
    \gamma_{\bK,A}\left(\widehat C_n^{\maltese,s}\right)=  n^{-1}\sum_{t=1}^n \prod_{j\in A}\left\{\cK_{j,F_{n}}(Y_{t+1-j}) -\mu_j\right\}.
\end{equation}
\begin{example}
For Spearman's rho, $K_j \equiv  D$, so $\cL_j(u) = \frac{u^2}{2}$.
For van der Waerden's coefficient, $K_j \equiv  \Phi$, so $\cL_j = -\phi\circ \Phi^{-1}$, $\mu_j=0$.
For Savage's coefficient, $K_j(x) \equiv  1-e^{-x}$, $x\ge 0$, so $\cL_j(u) = u-u\log{u}$, $\mu_j=1$, with the convention that $0\log{0}=0$.
Finally,  for the modified Blest's coefficient in the bivariate case, $K_1^{-1}(u)=u^2$, $u\in [0,1]$, $K_2=D$.
As a result, one gets the following formula for $12\gamma_{K_1,K_2}\left(\widehat C_n^\maltese\right)$:
$$
 2n^{-1}\sum_{i=1}^n \left\{ F_{nj}^2(X_{ij}-)  + F_{nj}(X_{ij}-)F_{nj}(X_{ij})+F_{nj}^2(X_{ij})-1\right\}
\left\{
F_{n2}(X_{i2}-)+F_{n2}(X_{i2}) -1\right\}.
$$
\end{example}
\begin{remark}
For continuous margins,
    $\Delta_{F_{nj}}(X_{ij})=n^{-1}$ a.s., so
    $$
   \cK_{j,F_{nj}}(X_{ij}) = \dfrac{\cL_j\{F_{nj}(X_{ij})\} -\cL_j\{F_{nj}(X_{ij}-)\} }{\Delta_{F_{nj}}(X_{ij})}\approx K_j^{-1}\left\{\frac{n}{n+1}F_{nj}(X_{ij})\right\}.
    $$
    Note that in general,  $n^{-1}\sum_{i=1}^n K_j^{-1}\left\{\frac{n}{n+1}F_{nj}(X_{ij})\right\}\neq \mu_j$, while
    $\disp
    n^{-1}\sum_{i=1}^n \cK_{j,F_{nj}}(X_{ij}) = \mu_j$, $j\in\setd$.
    This shows that even for continuous margins, one should use formulas \eqref{eq:TnAnonserial}--\eqref{eq:TnAserial} based on the multilinear copulas, since we do not need to work with the normalised $\frac{n}{n+1}F_{nj}(X_{ij})$.
  \end{remark}

The following result is an immediate consequence of Theorem \ref{thm:main}, the continuous mapping theorem,  together with representations \eqref{eq:TnAnonserial} and \eqref{eq:TnAserial}. When $K_j^{-1}$ is unbounded, one can use the same technique as in the corresponding proofs in \cite{Genest/Remillard:2004}, meaning that one integrates $\dG_{A,n}\{\bK(\bx)\}$ on large compact sets and show that the remainder can be made arbitrarily small, since $K_j^{-1}$ is square integrable by hypothesis, $K_j$ having finite variance. The covariance formulas follows from (D.6)-(D.7) in \cite{Nasri:2022}.

\begin{corollary}\label{cor:dep}
    Under the null hypothesis of independence, $\left\{\sn\gamma_{\bK,A}\left(\widehat C_n^\maltese\right): A\in \cN_d\right\}$ converge jointly to independent Gaussian random variables with variance $\disp \varsigma_{\bK,\bF,A}^2 = \prod_{j\in A}\varsigma_{K_j,F_j}^2$,
    where for any cdf $G$,
   \begin{equation}\label{eq:varsigma}
    \varsigma_{K_j,G}^2 = \int \left\{ \cK_{j,G}\{G(x)\} -\mu\right\}^2 dG(x)=\int_{\dR^2} \Gamma_{G}\{K_j(x),K_j(y)\} dxdy, \qquad j\in \setd.
   \end{equation}
Furthermore, under the null hypothesis of randomness, $\left\{\sn\gamma_{\bK,A}\left(\widehat C_n^{\maltese,s}\right): A\in \cS_d\right\}$ converge jointly to independent Gaussian random variables with variance $\disp \varsigma_{\bK,F,A}^2 =  \prod_{j\in A}\varsigma_{K_j,F}^2$.
\end{corollary}

\begin{remark}\label{rem:varsigma}
 It follows from  \cite{Genest/Remillard:2004} that  $\sigma_j^2 = \int_{\dR^2} \{K_j(x\wedge y)-K_j(x)K_j(y)\}dxdy$.   Finally,  $\varsigma_{K_j,F_j}^2 = \var\left\{\cK_{j,F_j}(X_j)\right\}$, if $X_j\sim F_j$, $j\in \setd$.
\end{remark}

The next result is fundamental  for applications since it shows how to normalised the statistics to standard Gaussian distributions in the limit.
Its proof is given in Appendix \ref{app:pf-lem}.

\begin{lem}\label{lem:varest} In the non-serial case,
$$
s_{K_j,F_{nj}}^2 = n^{-1}\sum_{i=1}^n \left[\dfrac{\cL_j\{F_{nj}(X_{ij})\} -\cL_j\{F_{nj}(X_{ij}-)\} }{\Delta_{F_{nj}}(X_{ij})}-\mu_j\right]^2 \stackrel{Pr}{\longrightarrow} \varsigma_{K_j,F_j}^2, \qquad j\in \setd,
$$
and in the serial case,
$$
s_{K_j,F_n}^2 = n^{-1}\sum_{t=1}^n \left[\dfrac{\cL_j\{F_{n}(Y_{t})\} -\cL_j\{F_{n}(Y_{t}-)\} }{\Delta_{F_{n}}(Y_{t})}-\mu_j\right]^2  \stackrel{Pr}{\longrightarrow} \varsigma_{K_j,F}^2 , \qquad j\in \setd.
$$
\end{lem}
From Corollary \ref{cor:dep} and Lemma \ref{lem:varest}, we obtain the next result, proven by \cite{Nasri:2022} in the serial case.
\begin{corollary}\label{cor:corre}
Under the null hypothesis of independence, $\sn r_{A,n} =  \sn \dfrac{\gamma_{\bK,A}\left(\widehat C_n^\maltese\right)}{\prod_{j\in A} s_{K_j,F_{nj}} }$, $A\in \cN_d$, converge jointly in law to independent standard Gaussian random variables. In addition, under the null hypothesis of randomness, $\sn r_{A,n} =  \sn \dfrac{\gamma_{\bK,A}\left(\widehat C_n^{\maltese,s}\right)}{\prod_{j\in A} s_{K_j,F_{n}} }$, $A\in \cS_d$, converge jointly in law to independent standard Gaussian random variables.
\end{corollary}


Finally, one can ask what happens when the null hypothesis of independence or randomness does not hold. Note that even when the margins were assumed to be continuous, \cite{Genest/Remillard:2004} did not  answer this question. Here, we provide an answer, under additional moments conditions.
\begin{ass}\label{hyp:momentsKj}
There exist $p_1,\ldots, p_d >1 $ such $\disp \sum_{j=1}^d \frac{1}{p_j}=1$ and
$\disp \int_0^1 \left|K_j^{-1}(u)\right|^{p_j} du <\infty$ for $j\in \setd$.
\end{ass}
Before stating the result, for any $A\in\cN_d$, in the non-serial case, set
\begin{equation}\label{eq:TnAnonserial1}
    \gamma_{\bK,A}\left(C_n^\maltese\right)  = n^{-1}\sum_{i=1}^n \prod_{j\in A}\left\{\cK_{j,F_{j}}(X_{ij})-\mu_j\right\}
\end{equation}
while in the serial case, for any $A\in \cS_d$, set
\begin{equation}\label{eq:TnAserial}
    \gamma_{\bK,A}\left(C_n^{\maltese,s}\right)=  n^{-1}\sum_{t=1}^n \prod_{j\in A}\left\{\cK_{j,F}(Y_{t+1-j}) -\mu_j\right\},
\end{equation}

\begin{corollary}\label{cor:H1} In the non-serial case, under Assumption \ref{hyp:momentsKj},  for any $A\in \cN_d$, both $\gamma_{\bK,A}\left(\widehat C_n^\maltese\right)$ and $\gamma_{\bK,A}\left(C_n^\maltese\right)$ converge in probability to
\begin{equation}\label{eq:gammagen}
 \gamma_{\bK,A}\left(C^\maltese\right) = E\left[ \prod_{j\in A} \left\{\cK_{j,F_j}(X_{ij})-\mu_j\right\}\right].
 \end{equation}
In the serial case, under Assumption \ref{hyp:momentsKj}, if the series $(U_t)_{t\ge 1}$ is stationary and ergodic, and $Y_t=F^{-1}(U_t)$, then for any $A\in \cS_d$, both  $\gamma_{\bK,A}\left(\widehat C_n^{\maltese,s}\right)$  and  $\gamma_{\bK,A}\left( C_n^{\maltese,s}\right)$  converge in probability to
\begin{equation}\label{eq:gammagenserial}
    \gamma_{\bK,A}\left(C^{\maltese,s}\right) = E\left[ \prod_{j\in A} \left\{\cK_{j,F}(Y_{t+1-j})-\mu_j\right\}\right], \quad t\ge d.
        \end{equation}
\end{corollary}

\begin{proof}
In the non-serial case, using the multinomial formula, one has
\begin{eqnarray*}
\gamma_{\bK,A} \left(\widehat C_n^\maltese \right) &=&
n^{-1} \sum_{i=1}^n  \left[  \prod_{j\in A} \left\{ \cK_{j,F_{nj}}(X_{ij})-\mu_j\right\}\right]\\
&=& \gamma_{\bK,A} \left(C_n^\maltese \right) + n^{-1} \sum_{i=1}^n  \sum_{B\subset A,\; B\neq \emptyset} \left[ \prod_{j\in A\setminus{B}} \left\{\cK_{j,F_j}(X_{ij})-\mu_j\right\} \right] \\
&&
\qquad \times   \left[ \prod_{j\in B} \left\{\cK_{j,F_{nj}}(X_{ij})-\cK_{j,F_j}(X_{ij})\right\}\right].
\end{eqnarray*}
The result will be proven if one can show that for any $j\in A$, as $n\to\infty$,
$$
n^{-1}\sum_{i=1}^n \left| \cK_{j,F_{nj}}(X_{ij})-\cK_{j,F_j}(X_{ij})\right|^{p_j}\stackrel{Pr}{\longrightarrow} 0.
$$
The latter follows from Lemma \ref{lem:moment}. The proof in the serial case is similar.
\end{proof}

If there is dependence, one can asks if there is a central limit theorem for $\gamma_{\bK,A}\left(\widehat  C_n^\maltese\right)$ whenever $|A|=2$.
If the supports of $K_1$ and $K_2$ are bounded and there exists an open set $\cO$ such that the partial derivatives $\partial_{u_j} C^\maltese$ exist and are continuous, $j\in \{1,2\}$, then Theorem 1 in \cite{Genest/Neslehova/Remillard:2017} yields that $\widehat  \dC_n^\maltese = \sn\left( \widehat  \dC_n^\maltese -C^\maltese \right)$ converges in $C(\cO)$ to a continuous centred Gaussian process $\widehat \dC^\maltese$, and
$$
\sn \left\{\gamma_{\bK,A}\left(\widehat  C_n^\maltese\right)-\gamma_{\bK,A}\left(C^\maltese\right) \right\} =
\sn \int \left[ \widehat  C_n^\maltese \left\{K_1(x_1),K_2(x_2)\right\} - C_n^\maltese \left\{K_1(x_1),K_2(x_2)\right\}\right]dx_1 dx_2
$$
converges in law to $\disp \int_{K_1^{-1}(\cO)} \int_{K_2^{-1}(\cO)} \widehat \dC^\maltese \{K_1(x_1),K_2(x_2)\}dx_1 dx_1$. Based on \cite{Remillard/Papageorgiou/Soustra:2012} and \cite{Nasri:2022}, a similar result should hold for $ \widehat  \dC_n^{\maltese,s}$, but the limiting distribution of the serial multilinear copula has not been studied yet but in the case when the $(Y_t)$s are iid \citep{Nasri:2022}. Also, one might conjecture that with appropriate moment conditions, the central limit theorem should hold for $K_j$ with unbounded support.

\section{Asymptotic behaviour along contiguous alternatives and local power}\label{sec:contiguity}
In this section, we consider contiguous alternatives of the form $C_{\btheta_n}$, where $C_{\btheta_0} = \Pi$, for some $\theta_0$,  and $\btheta_n = \btheta_0+ \sni \bdelta$, \BB{ where $\btheta_0$ and $\bdelta$ are column vectors in $\bR^p$.}
In the serial case, it is assumed that the sequence $(U_t)_{t\ge 1}$ is
 a $d$-Markov process with copula $C_{\btheta_n}$, meaning that the distribution of $(U_t,\ldots, U_{t+1-d})$ is $C_{\btheta_n}$, with density $c_{\btheta_n}$, and the joint density of $(U_1,\ldots, U_n)$ at
  $(u_1,\ldots,u_n)$   is
\begin{equation}\label{eq:markov}
c_{d-1,\theta_n}(u_{d-1},\ldots, u_1) \times  \prod_{t=d}^n\frac{c_{\btheta_n}(u_t,u_{t-1},\ldots,u_{t+1-d})}{ c_{d-1,\btheta_n}(u_{t-1},\ldots,u_{t+1-d})},
\end{equation}
where $\disp c_{d-1,\btheta_n}(u_{d-1},\ldots, u_1) = \int_0^1 c_{\btheta_n}(s,u_{d-1},\ldots,u_1)ds$.
It is assumed that the copula family $C_\btheta$ is smooth enough, namely that the Conditions 1--2 in \cite{Genest/Neslehova/Remillard/Murphy:2019} are met. More precisely, \BB{these conditions are that}
 $C_\btheta$ has a continuous density $c_\btheta$  continuously differentiable  with square integrable gradient $\dot c_\btheta$ in a neighbourhood of $\btheta_0$, with $\dot c = \left. \nabla_\btheta c_\btheta (\bu)\right|_{\btheta=\btheta_0}$,  $\bu \in (0,1)^d$,  $\disp \dot C (\bu) = \int_{(0,\bu]} \dot c (\bs) d\bs$, and
\begin{equation}\label{eq:condvdVW}
\lim_{n\to \infty}\int_{(0,1)^d}  [ {n}^{1/2} [ \{  c_{\btheta_n} (\bu) \}^{1/2} -1] - \bdelta^\top \dot c
(\bu)/2  ]^2 d\bu = 0.
\end{equation}
\BB{Here, $\nabla_\btheta f_\btheta $ is the column vector with components $\partial_{\theta_j} f_\btheta$, $j\in \setp$. }
\BB{ Using the mean value theorem, one can see that the following stronger conditions implies \eqref{eq:condvdVW}:
\begin{eqnarray}\label{eq:condvdVW2a}
\lim_{n\to \infty}\int_{(0,1)^d} \sup_{\|\btheta-\btheta_0\|\le n^{-1/2}\|\bdelta\|} \left\|\dot c_\btheta(\bu)- \dot c
(\bu)\right\|^2 d\bu &=& 0,\\
\limsup_{n\to \infty}\int_{(0,1)^d} \sup_{\|\btheta-\btheta_0\|\le n^{-1/2}\|\bdelta\|} \left\|\dot c_\btheta(\bu)\right\|^4 d\bu &<& \infty.
\label{eq:condvdVW2b}
\end{eqnarray}
As exemplified in the Appendix \ref{app:supp}, the latter conditions are met for several bivariate copula families, including the Gaussian, Farlie-Gumbel-Morgenstern, Clayton and Frank. However they do not hold for Gumbel's copula since $\dot c$ is not square integrable.}
Before stating the limiting distribution under the sequence of contiguous alternatives $C_{\btheta_n}$, for any $A\in\cN_d$, set $q_A = \cM_A(\dot C)$.
It follows from Lemma \ref{lem:commut}, stated in the Appendix, and proven in \cite{Nasri:2022}, that in the non-serial case,
$ \disp \dM_{\bF}(q_A)= \cM_A \circ \dM_\bF\left(\dot C\right) = \cM_A\left(\dot C^\maltese\right)$, while in the serial case,
$\disp
\dM_{F^{\otimes d}}(q_A)= \cM_A \circ \dM_{F^{\otimes d}}\left(\dot C\right) = \cM_A\left(\dot C^{\maltese,s}\right)$.
Under the previous conditions, the following results were obtained by \cite{Genest/Neslehova/Remillard/Murphy:2019} in the non-serial case, and by \cite{Nasri:2022} in the serial case.

\begin{thm}\label{thm:contiguous}
Under the sequence of contiguous alternatives $C_{\btheta_n}$,  in the non-serial case, the processes $\dG_{A,n}^\maltese$, $A\in\cN_d$, converge jointly in $\ell^\infty\left((0,1)^d\right)$ to  $\dG_A^\maltese +  \bdelta^\top  \cM_A(\dot C^\maltese)$.  Furthermore, in the serial case, the processes $\dG_{A,n}^{\maltese,s}$, $A\in\cS_d$, converge jointly
 in $\ell^\infty\left((0,1)^d\right)$ to  $\dG_A^{\maltese,s} +  \bdelta^\top \cM_A(\dot C^\maltese)$.
\end{thm}
\begin{remark}
   \cite{Nasri:2022} also considered Poisson contiguous alternatives with conditional mean $\lambda_{t,n} = \lambda_0+ \delta \sni Y_{t-1}$. In this case,  for any $A\in\cS_d$, the processes $\dG_{A,n}^{\maltese,s}$ converge jointly
 in $\ell^\infty\left((0,1)^d\right)$ to  $\dG_A^{\maltese,s} +  \frac{\delta}{\lambda_0} \I\{A=\{1,2\}\}\dM_F(f)(u_1) \dM_F(f)(u_2)$, where $f(u) =  \left\{\cL_F(u)-\lambda_0 u\right\}$, and $F$ is the cdf of the Poisson with parameter $\lambda_0$.
\end{remark}
As a corollary, we obtain the asymptotic behaviour of the proposed statistics for testing independence or randomness  under  the sequence of contiguous alternatives $C_{\btheta_n}$.

\begin{corollary} \label{cor:contig}
  Under the sequence of contiguous alternatives $C_{\btheta_n}$,  in the non-serial case,  the random variables $\sn \gamma_{\bK,A}\left(\widehat C_{n}^\maltese\right)$, $A\in\cN_d$, converge jointly
  to  independent Gaussian random variables with mean  $\bdelta^\top  \; \dot \gamma_{\bK,A}\left(C^\maltese\right)$ and variance $\varsigma_{\bK,\bF,A}^2$,
  where
\begin{equation}\label{eq:meannonserial}
  \dot \gamma_{\bK,A}\left(C^\maltese\right) = \int \dot c_A(\bu)  \prod_{j\in A}\left\{\cK_{j,F_j}\circ F_j^{-1}(u_j)-\mu_j\right\} d\bu,
\end{equation}
and $C_A$ is the copula restricted to components $U_j$ with $j\in A$.
  Furthermore, in the serial case,  the random variables $\sn \gamma_{\bK,A}\left(\widehat C_{n}^{\maltese,s}\right)$, $A\in\cS_d$, converge jointly
  to  independent Gaussian random variables with mean  $\bdelta^\top   \;  \dot \gamma_{\bK,A}\left(C^{\maltese,s}\right)$ and variance $\varsigma_{\bK,F,A}^2$,
  where
  \begin{equation}\label{eq:meanserial}
  \dot \gamma_{\bK,A}\left(C^{\maltese,s}\right) = \int \dot c_A(\bu)  \prod_{j\in A}\left\{\cK_{j,F}\circ F^{-1}(u_j)-\mu_j\right\} d\bu.
\end{equation}
\end{corollary}
Note that if the margin $F_j$ is continuous, $\cK_{j,F_j}\circ F_j^{-1}(u_j) = K_j^{-1}(u_j)$. In particular, in the serial case, if the margin $F$ is continuous, then
$\disp  \dot \gamma_{\bK,A}\left(C^{\maltese,s}\right) = \int \dot c_A(\bu)  \prod_{j\in A}\left\{K_j^{-1}(u_j)-\mu_j\right\} d\bu$.
\begin{remark} Since
$\disp
\cM_A\left(C_\btheta^\maltese\right)(\bu) =  E_\btheta\left[  \prod_{j\in A}\left\{J_{F_j}(X_j,u_j)-u_j\right\}\right]$,
Proposition \ref{prop:DG} in Appendix \ref{app:aux} yields
\begin{eqnarray*}
    \gamma_{\bK,A}\left(C_\btheta^\maltese\right) &=&  (-1)^{|A|} \int_{\dR^A} E_\btheta\left[  \prod_{j\in A}\left[ J_{F_j}\{X_j,K_j(x_j)\} - K_j(x_j)\right]\right]d\bx = E_\btheta\left[  \prod_{j\in A}\left\{\cK_{j}(X_j)-\mu_j\right\}\right],
\end{eqnarray*}
so  $\disp \dot \gamma_{\bK,\bF,A} =  \left. \partial_\btheta  \; \gamma_{\bK,A}\left(C_\btheta^\maltese\right) \right|_{\btheta = \btheta_0}$.
As a result,  one obtains formulas \eqref{eq:meannonserial} and \eqref{eq:meanserial}.
In particular, if $\disp \dot c_A = \prod_{j\in A} J_j(u_j)$, then in the non-serial case, for any $A\in\cN_d$,
\begin{equation}\label{eq:localpowernonserial}
\dot \gamma_{\bK,A} \left(C^\maltese\right) = \prod_{j\in A} \int_0^1 \left\{\cK_{j}\circ F_j^{-1}(u_j)-\mu_j\right\}J_j(u_j) du_j = \prod_{j\in A} \cov\left\{\cK_{j}\circ F_j^{-1}(U), J_j(U)\right\} ,
\end{equation}
where $U\sim \unif$, while in the serial case, for any $A\in\cS_d$,
\begin{equation}\label{eq:localpowerserial}
\dot \gamma_{\bK,A} \left(C^{\maltese,s}\right) = \prod_{j\in A} \cov\left\{\cK_{j}\circ F^{-1}(U), J_j(U)\right\}.
\end{equation}
\end{remark}

\subsection{Applications for local power}

First note that for many copula families satisfying the smoothness conditions listed at the beginning of the section, one has $\dot c_A(\bu) = \prod_{j\in A} J(u_j)$, and $J_j$ is often a quantile function. In this case, choosing $K_j^{-1}=J_j$ would make sense in order to have a non-zero mean, and hence having more local power by maximising formulas
\eqref{eq:localpowernonserial}--\eqref{eq:localpowerserial}. This is what was proposed in \cite{Genest/Verret:2005} in the bivariate case, where the margins were assumed to be continuous. In fact, the next proposition shows that this choice is also optimal for any margins. The proof of the following result is given in Appendix \ref{app:pf-are}.

\begin{prop}\label{prop:are} Suppose that  $\dot c_A(\bu) \propto \prod_{j\in A} G_j^{-1}(u_j)$, $\bu\in (0,1)^d$, where $\bG=(G_1,\ldots, G_d)$ is a vector of margins with means $\left(\tilde\mu_1,\ldots,\tilde \mu_d\right)$ and variances $\left(\tilde\sigma_1^2,\ldots,\tilde\sigma_d^2\right)$, and assume $U\sim \unif$. Then, in the non-serial case
$\disp \cov\left\{\cK_{j}\circ F_j^{-1}(U), G_j^{-1}(U)\right\} =\cov\left\{\cK_{j}(X_j), \cG_{j}(X_j)\right\} $, $j\in\setd$, where $X_j\sim F_j$, so
$\dot \gamma_{\bK,A} \left(C^\maltese\right)= \prod_{j\in A} \cov\left\{\cK_{j}(X_j), \cG_{j}(X_j)\right\}$. In particular, if $\bK=\bG$, then
 $\dot \gamma_{\bK,A} \left(C^\maltese\right)= \varsigma_{\bK,\bF,A}^2$. In the serial case,
$ \disp \cov\left\{\cK_{j}\circ F^{-1}(U), G_j^{-1}(U)\right\}=\cov\left\{\cK_{j}(X), \cG_{j}(X)\right\}$, $j\in\setd$, where $X\sim F$, so $\dot \gamma_{\bK,A} \left(C^{\maltese,s}\right)= \prod_{j\in A}
\cov\left\{\cK_{j}(X), \cG_{j}(X)\right\}$. In particular, if $\bK=\bG$, then
$\dot \gamma_{\bK,A} \left(C^{\maltese,s}\right)= \varsigma_{\bK,F,A}^2$.
\end{prop}

\begin{remark}\label{rem:ARE}
 Under the assumptions of Proposition \ref{prop:are}, it follows from
 Proposition 3 in \cite{Genest/Verret:2005} that
 the ARE between the test based on $\bK$, and $\bG$ is given by
 $\disp  \prod_{j\in A}\cor^2\left\{\cK_{j}(X_j), \cG_{j}(X_j)\right\}$
 in the non-serial case, and the ARE is
  $ \disp  \prod_{j\in A}\cor^2\left\{\cK_{j}(X), \cG_{j}(X)\right\}$
 in the serial case. This shows that whenever $\dot c_A(\bu) \propto \prod_{j\in A} G_j^{-1}(u_j)$, the ARE is maximised by taking $\bK=\bG$. Moreover, this result is independent of the margins, although the solution might not be unique.
 This is the case for example for a Bernoulli margin in the serial case. In fact, for any $A\in\cS_d$, using the coefficients $r_{A,n}$ defined in Lemma \ref{lem:varest}, one gets that
 $n r_{A,n}^2 = \frac{\dZ_{A,n}^2}{\{p_n(1-p_n)\}^{|A|}}$, where $\disp p_n = n^{-1}\sum_{t=1}^n \{Y_t=1\}$, and $\disp \dZ_{A,n} = \sni\sum_{t=1}^n \prod_{j\in A}
 \left[ \I\{Y_{t+1-j}=1\}-p_n\right]$. In the non-serial case, if all margins are Bernoulli, $\disp p_{nj} = n^{-1}\sum_{i=1}^n \{X_{ij}=1\}$, and $\disp \dZ_{A,n} = \sni\sum_{i=1}^n \prod_{j\in A}
 \left[ \I\{X_{ij}=1\}-p_{nj}\right]$, $A\in \cN_d$, then $nr_{A,n}^2 = \frac{\dZ_{A,n}^2}{\prod_{j\in A}p_{nj}(1-p_{nj})}$.
\end{remark}

We now consider some copula families studied in the continuous case by  \cite{Genest/Quessy/Remillard:2007}.
\begin{example}\label{ex:are}
If $C_\theta$ is the equicorrelated Gaussian copula, then
$\disp  \dot c_A(\bu) = \sum_{B\subset A, |B|=2}\prod_{j\in B}\Phi^{-1}(u_j)$. It follows from \eqref{eq:localpowernonserial}--\eqref{eq:localpowerserial}
that in the non-serial case,
$\disp
\dot \gamma_{\bK,A} \left(C^\maltese\right) = \I\{|A|=2\} \prod_{j\in A} \cov\left\{\cK_{j}\circ F_j^{-1}(U), \Phi^{-1}(U)\right\}$,
while in the serial case,
$\disp
\dot \gamma_{\bK,A} \left(C^{\maltese,s}\right)= \I\{|A|=2\} \prod_{j\in A}\cov\left\{\cK_{j}\circ F^{-1}(U), \Phi^{-1}(U)\right\}$.
As a result, van der Waerden's coefficients should be locally the most powerful when restricted to pairs, i.e., when $|A|=2$. This is not surprising since it coincides with  Pearson's correlation between  $\Phi^{-1}(U_1)$ and  $\Phi^{-1}(U_2)$, if $(U_1,U_2)\sim C_\theta$.

For the Farlie-Gumbel-Morgensten's copula family,
$\disp \dot c_A(\bu) = \I\{A=\setd\}\prod_{j=1}^d (1-2u_j)$.
It follows that in the non-serial case,
$\disp
\dot \gamma_{\bK,A} \left(C^\maltese\right) = 2^d (-1)^d \I\{A=\setd\} \prod_{j=1}^d \cov\left\{\cK_{j}\circ F_j^{-1}(U), U\right\}$,
and in the serial case,
$\disp
\dot \gamma_{\bK,A} \left(C^{\maltese,s}\right)= 2^d (-1)^d \I\{A=\setd\} \prod_{j=1}^d  \cov\left\{\cK_{j}\circ F^{-1}(U), U\right\}$,
so Spearman's rho with $A=\setd$ should be locally the most powerful.

For Claytons's copula family,
$\disp \dot c_A(\bu) = \sum_{B\subset A, |B|=2}\prod_{j\in B}(1+\log{u_j})$. In the non-serial case, one gets
$\disp
\dot \gamma_{\bK,A} \left(C^\maltese\right) = \I\{|A|=2\} \prod_{j\in A} \cov\left\{\cK_{j}\circ F_j^{-1}(U), \log{U}\right\}$,
and in the serial case,
$\disp
\dot \gamma_{\bK,A} \left(C^{\maltese,s}\right)= \I\{|A|=2\} \prod_{j\in A}\cov\left\{\cK_{j}\circ F^{-1}(U), \log{U}\right\}$.
As a result, Savage's coefficients for pairs should be locally the most powerful.

Finally, for Frank's copula family,
$\disp
\dot c_A(\bu) = \frac{|A|-1}{2} + 2^{|A|-1}\prod_{j\in A}u_j - \sum_{j\in A} u_j$, and it then follows from formulas \eqref{eq:meannonserial}--\eqref{eq:meanserial}
that in the non-serial case,
$
\dot \gamma_{\bK,A} \left(C^\maltese\right) = 2^{|A|-1}  \prod_{j\in A} \cov\left\{\cK_{j}\circ F_j^{-1}(U), U\right\}$,
and in the serial case,
$
\dot \gamma_{\bK,A} \left(C^{\maltese,s}\right)=  2^{|A|-1}  \prod_{j\in A} \cov\left\{\cK_{j}\circ F^{-1}(U), U\right\}$.
So even if $\dot c_A$ is not a product,  $\dot \gamma_{\bK,A} $ can be computed. As a result, Spearman's rho for all sets should be locally the most powerful. The good performance of combination of Spearman's rho for pairs was confirmed in numerical experiments in the serial case for Frank's family; see, e.g., \cite{Nasri:2022}.
\end{example}

\subsection{Wald's test statistics}\label{ssec:combination}

As stated in Corollary \ref{cor:corre}, in the non-serial case, the limiting distributions of the statistics $\sn r_{A,n}$ are independent.  \cite{Littell/Folks:1971,Littell/Folks:1973} showed that, in terms of Bahadur's efficiency, the ``best'' test  in this case is given by
$\disp -2\sum_{A\in\cN_d}\log\left\{2-2\Phi\left(\sn|r_{A,n}|\right)\right\}$,  where $\Phi$ is the cdf of the standard Gaussian distribution. However, given that we have independent standard Gaussian limits, there is not much difference in terms of power by using the Wald's statistic, i.e.,  the sum  of squared statistics $\disp L_{n,p,d} = n \sum_{A\subset \cN_d, |A|\le p}r_{A,n}^2$ or $\disp L_{n,p,d} = n \sum_{A\subset \cS_d, |A|\le p}r_{A,n}^{2}$.
This was shown numerically in \cite{Nasri:2022}.
In the non-serial case, one could consider all sets $A\in \cN_d$, so $L_{n,d,d}$ has approximately a chi-square distribution with $2^d-d-1$ degrees of freedom, or consider only the pairs, i.e., $L_{n,2,d}$, which has approximately  a chi-square distribution with $\frac{d(d-1)}{2}$ degrees of freedom. In the serial case,
$L_{n,d,d}$ has approximately a chi-square distribution with $2^{d-1}-1$ degrees of freedom, while $L_{n,2,d}$, which has approximately  a chi-square distribution with $d-1$ degrees of freedom. One can also draw dependograms, i.e., graphs of $\sn r_{A,n}$ plotted as a functions of all possible sets $A$ or all pairs. These statistics are implemented in the CRAN package \textit{MixedIndTests} \citep{Nasri/Remillard/Neslehova/Genest:2022}.

 \section{Numerical experiments}\label{sec:num}

In what follows, we consider only the serial case \BB{with $d=5$} and the following copula families: independence, Tent map, Farlie-Gumbel-Morgenstern (FGM) (with $\theta=1$), and Clayton, Frank and Gaussian families with Kendall's tau of $0.1$.
Recall that the Tent map copula is the joint cdf of $(U_1, 2\min(U_1,1-U_1))$, with $U_1\sim \unif$. The generated series are all stationary and Markov, with the exception of the FGM which is $2$-Markov, as defined by \eqref{eq:markov}. We consider the same set of $7$ margins as in \cite{Nasri:2022}, namely $F_1$ is Bernoulli with $p=0.8$, $F_2$ is Poisson(6), $F_3$ is a negative binomial NB(r=1.5,p=0.2), $F_4$ is a mixture of $0$ with probability  $0.1$ and Poisson(10) otherwise, $F_5$ is a mixture of $0$ with probability  $0.1$ and N(0,1) otherwise, $F_6$ is a discretized Gaussian with $F_6^{-1}(u)= \left\lfloor 200 \Phi^{-1}(u)\right\rfloor$, and $F_7$ is a discrete Pareto with $F_7(k) = 1-\frac{1}{k+1}$, $k\in\dN$.
For the tests, we considered the statistics $L_{n,2,5}$ and $L_{n,5,5}$ for Spearman, van der Waerden, and Savage coefficients, for $n\in\{100,250,500\}$. The simulations results, based on $N=1000$ replications, appear in Table \ref{tab:ind+tent} for the independence  and the Tent map copulas, in Table \ref{tab:fgm+clayton} for the Farlie-Gumbel-Morgenstern copulas, and in Table \ref{tab:gaussian+frank} for the Gaussian and Frank copulas.

From the results for the independence copula in Table \ref{tab:ind+tent}, the empirical levels of the tests are quite satisfactory, being close to the 5\% target. Next, for the Tent map copula, Savage's test is surprisingly good, compared to the two other coefficients, with the exception of the Bernoulli margin $F_1$ which give the same results for all tests. The good performance of Savage's test might come from the fact that for continuous margins, the theoretical coefficient is not $0$, contrary to Spearman's rho and van der Waerden coefficients \citep{Remillard:2013}. Next, from Table \ref{tab:fgm+clayton}, without any surprise, the tests based on $L_{n,2,5}$
are not powerful for the Farlie-Gumbel-Morgenstern copula, given the calculations in Example \ref{ex:are}, while the best test is $L_{n,5,5}$ based on Spearman's rho, as predicted. Also from the computations in Example \ref{ex:are}, the tests based on Savage's coefficients are the best for the Clayton's copula. Finally, from the results in Table \ref{tab:gaussian+frank}, as predicted, the tests based on Spearman's rho are the best for Frank's copula, while  the tests based on van der Waerden's coefficient are the best for the Gaussian copula. These results all agree with the results in Example \ref{ex:are}, as well as the results of \cite{Genest/Verret:2005} for the bivariate case with continuous margins.

\section{Conclusion}
For the non-serial and serial settings, we defined tests of independence and tests of randomness derived from bivariate and multivariate extensions of several known copula-based dependence measures that are usually defined for observations with continuous distributions. \BB{We showed that working with observations with arbitrary distributions did not add computational difficulties, and even if the statistics are not margin-free, the simulation results proved that the power was quite good under alternative hypotheses.} We also deduced the locally most powerful tests based  on proposed statistics for some known copula families, whatever the margins. These results generalise the previous findings of \cite{Genest/Verret:2005} in the bivariate when the margins were assumed to be continuous. \BB{In future work, it would be interesting to study extensions of $\phi$-dependence measures for arbitrary data using the multililinear copula density.}

 \appendix
 \section{Auxiliary results} \label{app:aux}

 Here we define two important transformations: the M\"obius transform and the multilinear interpolation.
For $A\in \cN_d$, the M\"obius transform  $\cM_A$ is defined by
$\disp
\cM_A(f)(\bu) = \sum_{B\subset A}(-1)^{|A\setminus B|} f\left(\bu^B\right)\prod_{j\in A\setminus B} u_j$,
where  $\bu^B \in [0,1]^d$ is such that
$\disp u^B_j = \left\{ \begin{array}{ll} u_j & \mbox{if } j \in B,\\
1 & \mbox{if } j \not \in B. \end{array} \right.$. In particular, if $f = f_1\otimes \cdots \otimes f_d$, i.e., $f(\bu)=\prod_{j=1}^d f_j(u_j)$, and $f_j(1)=1$, then $\cM_A(f) = \prod_{j=1}^d \{f_j(u_j)-u_j\}$. As a result, for any $A\in\cN_d$, $\cM_A(\Pi)\equiv 0$.
Next, following \cite{Genest/Neslehova/Remillard:2017}, for $\bF=(F_1,\ldots, F_d)$, we define the interpolation operator $\dM_{\bF}$. To this end, for arbitrary $\bu = (u_1,\ldots,u_d) \in [0,1]^d$ and $S \subseteq \{ 1, \ldots , d\}$, and for any $B\subset \setd$, set $(\bu_{F,B})_j = F_j\circ F_j^{-1}(u_j)$ if $j\in B$, and
$(\bu_{\bF,B})_j = F_j\left\{ F_j^{-1}(u_j)-\right\}$ if $j\notin B$. In particular, if $F_j$ is continuous at $F_j^{-1}(u_j)$, then $(\bu_{\bF,B})_j = u_j$ for any $B$. Note that $\bu_{\bF,S}$ is an element in the closure $\bar \cR_{\bF}$ of $\cR_{\bF} = \cR_{F_1} \times \cdots \times \cR_{F_d}$.
Further let $\ell_\infty (K)$ be the collection of bounded real-valued functions on $K \subseteq [0,1]^d$.  The multilinear interpolation operator $\dM_{\bF}$, is then defined for all $g \in \ell_\infty (  \bar \cR_{\bF} )$ and $\bu \in [0,1]^d$, by
$\disp \dM_{\bF}(g)(\bu) = \sum_{B\subset \setd} g(\bu_{\bF,B}) \left\{ \prod_{j \in B} \lambda_{F_j} (u_j) \right\} \left[ \prod_{j \in B^\complement} \{ 1-\lambda_{F_j}(u_j) \} \right]$.
In particular, if $\disp g(\bu) = \prod_{j=1}^d g_j(u_j)$, then
$\disp \dM_{\bF}(g)(\bu) = \prod_{j=1}^d \dM_{F_j}(g_j)(u_j)$.
The following commutation result was proven in \cite{Nasri:2022}.

 \begin{lem}\label{lem:commut} For any $f=f_1\otimes \cdots \otimes f_d$, such that $f_{j}(1)=1$, and for any $A\in \cN_d$, one has
$$
\cM_A\circ \dM_{\bF}(f) = \dM_{\bF} \circ \cM_A(f).
$$
\end{lem}

The next result is fundamental for the computations of the proposed statistics.
 \begin{prop}
    \label{prop:DG}
    For any cdf $G$ with mean $\mu$ and variance $\sigma^2$, we have
\begin{equation}\label{eq:DG0}
 \int_{-\infty}^\infty \left[ J_F\{x,G(y)\}-G(y)\right]dy    = \mu - \cG_F(x),
\end{equation}
where $\cG_F(x)  = G^{-1}\{ F(x)\} $, if $F$ is continuous at $x$,  and
   $\cG_F(x)  = \dfrac{\cL_{G}\{F(x)\} -\cL_{G}\{F(x-)\} }{\Delta_F(x)}$,
if $F$ is not continuous at $x$, with $\disp \cL_G(u) = \int_0^u  G^{-1}(v)dv$.
\end{prop}
\begin{proof}
First, since $0\le J_F\le 1$, $\int J_F(x,u)dF(x)=u$, and $Y\sim G$ is integrable, it follows that
$\disp
\int \left[ \int_{-\infty}^\infty \left[ J_F\{x,G(y)\}-G(y)\right]dy \right]dF(x)  = 0$.
Next, for any $c\in\dR$,
$\disp E\left[ Y\I\{Y>c\}\right] = \int_c^\infty \bar G(y)dy -\max(0,-c)\bar G(c)$
and
$\disp
E\left[ Y\I\{Y\le c\}\right] = c -\int_{-\infty}^c   G(y)dy + \max(0,-c)\bar G(c)$.
As  a result,
\begin{equation}\label{eq:murep}
\mu = c+ \int_c^\infty \bar G(y)dy - \int_{-\infty}^c   G(y)dy.
\end{equation}
Set $a=F(x-)$ and $b=F(x)$. Further set $\bar G(y)=1-G(y)$, $y_0=G^{-1}(a)$ and
$y_1=G^{-1}(b)$.
Now, suppose that $\Delta_F(x)=b-a=0$.
Then, according to \eqref{eq:murep}
$$
-I = -\int_{y_1}^\infty\bar G(y)dy + \int_{-\infty}^{y_1} G(y)dy =  y_1-\mu = \cG_F(x)-\mu.
$$
Suppose now that $f(x) = b-a>0$.
Then,
\begin{eqnarray*}
I &=&   \int_{-\infty}^\infty \left[ D\left\{\frac{G(y)-F(x-)}{f(x)}\right\}-G(y)\right]dy  =
- \int_{-\infty}^{y_0}    G(y)dy  +\int_{y_0}^{y_1} \left[ \left\{\frac{G(y)-a}{b-a}\right\}-G(y)\right]dy\\
&& \qquad +  \int_{y_1}^{\infty} \bar G(y)dy \\
&=&  \int_{y_1}^{\infty} \bar G(y)dy - \int_{-\infty}^{y_1} G(y)dy  +\int_{y_0}^{y_1} \left\{\frac{G(y)-a}{b-a}\right\}dy
= \mu -y_1 + \int_{y_0}^{y_1} \left\{\frac{G(y)-a}{b-a}\right\}dy,
\end{eqnarray*}
using \eqref{eq:murep}. Finally,
\begin{multline*}
\int_{y_0}^{y_1} \{G(y)-a\}dy =  (y_1-y_0)\{G(y_0)-a\} + E\left[  \int_{y_0}^{y_1} \I\{y_0 < Y \le y\} dy \right] \\
= (y_1-y_0)\{G(y_0)-a\} +
E\left[  (y_1- Y) \I\{y_0 < Y \le y_1\} \right] \\
= (y_1-y_0)\{G(y_0)-a\} + y_1\left\{G(y_1)-G(y_0)\right\} -\int_a^b G^{-1}(v)dv  +  \int_a^b G^{-1}(v)dv -\int_{G(y_0)}^{G(y_1)} G^{-1}(v)dv \\
=    (y_1-y_0)\{G(y_0)-a\} + y_1\left\{G(y_1)-G(y_0)\right\}  - \cL_G(b)+\cL_G(a)  +\int_a^{G(y_0)} G^{-1}(v)dv -\int_b^{G(y_1)} G^{-1}(v)dv \\
= y_1(b-a) - \cL_G(b)+\cL_G(a).
\end{multline*}
As a result, $- I = \dfrac{\cL_G(b)-\cL_G(a)}{b-a}  -\mu = \cG_F(x)-\mu$.
\end{proof}

The next two results are used to determine the limiting value of the statistics $\gamma_{\bK,A}\left(\widehat C_n^{\maltese}\right)$ and
$\gamma_{\bK,A}\left(\widehat C_n^{\maltese,s}\right)$ where the null hypothesis of independence or randomness is not satisfied. First, for $ \bu\in [0,1]^d$, set
$\disp
C_A^\maltese(\bu) = \cM_A\left(C^\maltese\right)(\bu) = E\left[\prod_{j\in A}\left\{J_{F_j}(X_j,u_j)-u_j\right\}\right]$,
and
$\disp
C_A^{\maltese,s}(\bu) =  \cM_A\left(C^{\maltese,s}\right)(\bu) = E\left[\prod_{j\in A}\left\{J_{F}(Y_{t+1-j},u_j)-u_j\right\}\right]$.

\begin{lem}\label{lem:H1}
$\disp \sup_{\bu\in [0,1]^d}\left|\sni\dG_{n,A}^\maltese(\bu)- C_A^\maltese(\bu)\right|\to 0$ a.s., for any $A\in \cN_d$. If the series $U_t$ is stationary and ergodic, and $Y_t=F^{-1}(U_t)$, then $\disp \sup_{\bu\in [0,1]^d}\left|\sni\dG_{n,A}^{\maltese,s}(\bu)- C_A^{\maltese,s}(\bu)\right|\to 0$ a.s., for any $A\in \cS_d$.
\end{lem}

\begin{proof}
It follows from Lemma A.4 in \cite{Genest/Neslehova/Remillard:2017} that for any stationary and ergodic sequence $(Z_i)_{i\ge 1}$ with distribution function $G$,
$\disp \sup_{u\in [0,1]}\left|J_{G_n}(Z_i,u)-J_{G}(Z_i,u)\right|\to 0$ a.s., where $\disp G_n(z) = n^{-1}\sum_{i=1}^n \I(Z_i\le z)$. Using the multinomial formula, it then follows that for any $A\in\cN_d$,
$$
\sup_{\bu\in [0,1]^d}\left| \sni \dG_{A,n}^\maltese(\bu) - n^{-1}\sum_{i=1}^n \prod_{j\in A}\left\{ J_{F_j}(X_{ij},u_j)-u_j\right\}  \right| \to 0 \qquad a.s.
$$
To complete the proof in the non-serial case, it suffices to remark that
$$
\sup_{\bu\in [0,1]^d}\left| n^{-1}\sum_{i=1}^n \prod_{j\in A}\left\{ J_{F_j}(X_{ij},u_j)-u_j\right\}  -\cM_A\left(C^\maltese\right)(\bu)\right| \to 0 \qquad a.s.
$$
Finally, in the serial case,  using again the multinomial formula, it follows that for any $A\in\cS_d$,
$$
\sup_{\bu\in [0,1]^d}\left| \sni \dG_{A,n}^{\maltese,s}(\bu) - n^{-1}\sum_{i=1}^n \prod_{j\in A}\left\{ J_{F}(Y_{t+1-j},u_j)-u_j\right\}  \right| \to 0 \qquad a.s.
$$
To complete the proof, it suffices to remark that from the ergodic theorem,
$$
\sup_{\bu\in [0,1]^d}\left| n^{-1}\sum_{i=1}^n \prod_{j\in A}\left\{ J_{F}(Y_{t+1-j},u_j)-u_j\right\}  -\cM_A\left(C^{\maltese,s}\right)(\bu)\right| \to 0 \qquad a.s.
$$
\end{proof}

\begin{lem}\label{lem:moment}
Set $X=F^{-1}(U)$, with $U\sim\unif$, set $Z = G^{-1}(U)$, and suppose that $E\left[\left|Z\right|^p \right]<\infty$ for a given $p\ge 1$. Then
$\disp E\left[ \left| \cG_F(X) \right|^p \right]\le E\left[\left|Z\right|^p \right]$. In addition, if $U_1,\ldots, U_n$ are iid and $\unif$ and
$X_i=F^{-1}(U_i)$, then $\disp
\frac{1}{n}\sum_{i=1}^n \left|\cG_{F_n}(X_i)-\cG_F(X_i)\right|^p $
converges in probability to $0$, where $F_n(x) = \frac{1}{n}\sum_{i=1}^n \I(X_i\le x)$.
 Furthermore, if $(U_t)_{t\ge 1}$ is a stationary and ergodic sequence uniformly distributed on  $(0,1)$ and
$Y_t = F^{-1}(U_t)$, then
$\disp \frac{1}{n}\sum_{t=1}^n \left|\cG_{F_n}(Y_t)-\cG_F(Y_t)\right|^p$
converges in probability to $0$, if $F_n(y) = \frac{1}{n}\sum_{t=1}^n \I(Y_t\le y)$.
\end{lem}
\begin{proof}
 If $\cA$ is the set of atoms of $F$, $f(x)=F(x)-F(x-)$, $x\in\cA$,  and $\cB = \bigcup_{x\in\cA}(F(x-),F(x)]$, then
\begin{equation}\label{eq:cGFallap}
E\left[ \left| \cG_F(X) \right|^p \right] = \sum_{x\in \cA} f(x) \left| \cG_F(x) \right|^p
+E\left[\I(U\not \in \cB)\left| G^{-1}(U) \right|^p \right].
\end{equation}
Next,
\begin{multline}\label{eq:cGFallap2}
\sum_{x\in \cA} f(x) \left| \cG_F(x) \right|^p = \sum_{x\in\cA}f(x)\left|\frac{1}{f(x)}\int_{F(x-)}^{F(x)} G^{-1}(u)du \right|^p \\
\le
\sum_{x\in\cA} \int_{F(x-)}^{F(x)} \left|G^{-1}(u)\right|^p du = E\left[\I(U\in \cB)\left| G^{-1}(U) \right|^p \right].
\end{multline}
Combining \eqref{eq:cGFallap}--\eqref{eq:cGFallap2}, one gets that
$\disp E\left[ \left| \cG_F(X) \right|^p \right]\le E\left[\left|Z\right|^p \right]$.
Next, for any $M>0$, set $G_M = \disp\int_{-\infty}^{-M} G(x)dx$ and $\bar G_M = \disp\int_M^{\infty} \{1-G(x)\}dx$.
By assumption, $G_M$ and $\bar G_M$ converge to $0$ as $M\to\infty$. Also,
\begin{equation}\label{eq:intMub}
\bar L(u,M)= \int_{M}^\infty \left[\I\left\{u \le G(x)\right\}-G(x)\right] dx= \bar G_M -\left\{G^{-1}(u)-M\right\}^+,
\end{equation}
\begin{equation}\label{eq:intMlb}
L(u,M) = \int_{-\infty}^{-M} \left[\I\left\{u \le G(x)\right\}-G(x)\right] dx= -G_M +\left\{-M-G^{-1}(u)\right\}^+.
\end{equation}
Suppose first that $U_1,\ldots, U_n$ are iid.
As a result, if $\disp \int_0^1 \left|K^{-1}(u)\right|^p du $ is finite, $p>1$,  then for any $M>0$, it follows from Lemma \ref{lem:H1} that
$\disp \frac{1}{n}\sum_{i=1}^n \left|\int_{-M}^M \left[J_{F_n}\{X_{i},G(x)\}- J_F\{X_i,G(x)\} \right]dx \right|^p$
converges in probability to $0$ as $n\to\infty$. Next, using \eqref{eq:intMub}, one gets
\begin{multline*}
\frac{1}{n}\sum_{i=1}^n \left| \int_{M}^\infty \left[ J_{F_{n}}\{X_{i},G(x)\}- J_F \{X_{i},G(x)\} \right]dx \right|^p \\
\le \frac{1}{n} \sum_{i=1}^n  \left| \int_0^1 \left[ G^{-1}\left\{ F_n(X_i-)+s\Delta F_n(X_i) \right\} -M\right]^+ ds \right|^p\\
 + \frac{1}{n}\sum_{i=1}^n \left|\int_0^1 \left[ G^{-1}\left\{F(X_i-)+s\Delta F(X_i)\right\} -M\right]^+ ds \right|^p.
\end{multline*}
Next,
\begin{multline*}
\frac{1}{n}\sum_{i=1}^n
\left|\int_0^1 \left[ G^{-1}\left\{ F_n(X_i-)+s\Delta F_n(X_i) \right\} -M\right]^+ ds \right|^p  \\
\le
\sum_{x\in\cA} \int_{F_n(x-)}^{F_n(x)} \left\{ G^{-1}(u) -M\right\}^p \I\{u>G(M)\}  du \\
 + \sum_{i=1}^n \I(U_i\not\in \cB)
\int_{B_n(U_i)-1/n}^{B_n + s/n} \left\{ G^{-1}(u) -M\right\}^p  \I\{u>G(M)\} du  \\
\le
\int_{G(M)}^\infty  \left|G^{-1}(u) -M \right|^p \I\{u>G(M)\} du  +  \sum_{i=1}^n \I(U_i\not\in \cB)
\int_{B_n(U_i)-1/n}^{B_n + s/n} \left\{ G^{-1}(u) -M\right\}^p  \I\{u>G(M)\} du.
\end{multline*}
Now, $\disp \int_{G(M)}^\infty  \left|G^{-1}(u) -M \right|^p \I\{u>G(M)\} du $ can be made arbitrarily small if $M$ is large enough. Also,
$$
E \left[ \sum_{i=1}^n \I(U_i\not\in \cB)
\int_{B_n(U_i)-1/n}^{B_n(U_i)} \left| G^{-1}(u) -M\right|^p \I\{u>G(M)\}  du \right] \le  \int_{G(M)}^1 \left| G^{-1}(u) -M\right|^p du,
$$
which can be made arbitrarily small if $M$ is large enough. Next,
$$
\frac{1}{n}\sum_{i=1}^n \left|\int_0^1 \left[ G^{-1}\left\{F(X_i-)+s\Delta F(X_i)\right\} -M\right]^+ ds \right|^p
$$
converges almost surely to
$\disp
E\left[ \left| \int_0^1 \left[ G^{-1}\left\{ F(X_1-)+s\Delta F(X_1)\right\} -M\right]^+ ds \right|^p \right] \le \int_0^1 \left| G^{-1}(u)\right|^p du$,
since $F(X_1-)+V\Delta F(X_1)\sim \unif(0,1)$ if $V\sim \unif$  and $V$ is independent of $X_1$.
An analogous result holds for $\disp \frac{1}{n} \sum_{i=1}^n \left| \int_{-\infty}^{-M} \left[ J_{F_{n}}\{X_{i},G(x)\}-J_{F}\{X_{i},G(x)\} \right]dx \right|^p$. As a result,
$\disp
\frac{1}{n}\sum_{i=1}^n \left|\cG_{F_n}(X_i)-\cG_F(X_i)\right|^p$
converges in probability to $0$.
To complete the proof, note that the previous inequalities also hold for a stationary ergodic sequence $(U_t)_{t\ge 1}$.
\end{proof}

\section{Proofs}

    \subsection{Proof of Lemma \ref{lem:varest}}\label{app:pf-lem}
Without loss of generality, drop the subscript $j$. Then, by Proposition \ref{prop:DG},
\begin{eqnarray*}
    s_n^2 &=& n^{-1}\sum_{i=1}^n
    \left[ \dfrac{\cL\{F_n(X_i)\} -\cL\{F_n(X_i-)\} }{\Delta_{F_n}(X_i)}-\mu \right]^2\\
   &=&  n^{-1}\sum_{i=1}^n \left[\int_{\dR} \left[\int_0^1 \I\left\{ F_{n}(X_{i}-)+ s \Delta_{F_{n}}(X_{i})\le  K(x)\right\}ds -K(x)\right]dx \right]^2.
\end{eqnarray*}
Next, let $M>0$ be given and choose $\delta\in (0,1)$ so that $K^{-1}(1-\delta)> M$ and $K^{-1}(\delta)< -M$. To prove the result, it suffices to show that if $M$ is large enough, and $\delta$ is small enough,
$$
s_{n,1,M}^2 = n^{-1}\sum_{i=1}^n \left[ \int_{-M}^M  \left[\int_0^1 \I\left\{ F_{n}(X_{i}-)+ s \Delta_{F_{n}}(X_{i})\le  K(x)\right\}ds -K(x)\right]dx \right]^2
$$
can be arbitrarily close to $s^2$, while
$$
s_{n,2,M}^2 = n^{-1}\sum_{i=1}^n \left[ \int_{M}^\infty  \left[\int_0^1 \I\left\{ F_{n}(X_{i}-)+ s \Delta_{F_{n}}(X_{i})\le  K(x)\right\}ds -K(x)\right]dx \right]^2
$$
and
$$
s_{n,3,M}^2 = n^{-1}\sum_{i=1}^n \left[ \int_{-\infty}^{-M}  \left[\int_0^1 \I\left\{ F_{n}(X_{i}-)+ s \Delta_{F_{n}}(X_{i})\le  K(x)\right\}ds -K(x)\right]dx \right]^2
$$
can be made arbitrarily small. First,  as $n\to\infty$, $s_{n,1,M}^2$ converges in probability to
$$
s_{1,M}^2 = E \left[ \int_{-M}^M  \left[\int_0^1 \I\left\{ F(X-)+ s \Delta_{F}(X)\le  K(x)\right\}ds -K(x)\right]dx \right]^2.
$$
Using similar arguments as in \cite{Genest/Remillard:2004}, $s_{1,M}^2\to s^2$ as $M\to\infty$.  Next,
$ s_{n,2,M}^2 =s_{n,2a,M}^2 +s_{n,2b,M}^2+s_{n,2c,M}^2$, where
\begin{multline*}
s_{n,2a,M}^2 =
 n^{-1}\sum_{i=1}^n I\{F_n(X_i-)\le 1-\delta,F_n(X_i)\ge \delta\}\\
 \times \left[ \int_{M}^\infty  \left[\int_0^1 \I\left\{ F_{n}(X_{i}-)+ s \Delta_{F_{n}}(X_{i})\le  K(x)\right\}ds -K(x)\right]dx \right]^2,\\
\end{multline*}
 \begin{multline*}
 s_{n,2b,M}^2 =  n^{-1}\sum_{i=1}^n I\{F_n(X_i-)>1-\delta\}\\
  \times \left[ \int_{M}^\infty  \left[\int_0^1 \I\left\{ F_{n}(X_{i}-)+ s \Delta_{F_{n}}(X_{i})\le  K(x)\right\}ds -K(x)\right]dx \right]^2,
  \end{multline*}
 \begin{multline*}
 s_{n,2c,M}^2= n^{-1}\sum_{i=1}^n I\{F_n(X_i)< \delta\}\\
  \times \left[ \int_{M}^\infty  \left[\int_0^1 \I\left\{ F_{n}(X_{i}-)+ s \Delta_{F_{n}}(X_{i})\le  K(x)\right\}ds -K(x)\right]dx \right]^2 .
\end{multline*}
Now $s_{n,2a,M}^2$ converges in probability to $s_{2a,M}^2$, which can be made arbitrarily small by taking $M$ large enough. Next, $s_{n,2c,M}=0$ since $K(\delta)\le -M$. Finally,
$$
s_{n,2b,M}^2 \le \sum_{i=1}^n \I\{F_n(X_i-)> 1-\delta\}\left[\int_M^\infty \{1-K(x)\}dx\right]^2,
$$
which can be made arbitrarily small since $1-K$ is integrable on $(0,\infty)$. The case of $s_{n,3,M}$ is similar.
\qed

\subsection{Proof of Proposition \ref{prop:are}}\label{app:pf-are}
 If $\cA_j$ is the set of atoms of $F_j$, and $\disp \cI_j = \cup_{x\in\cA_j}\left(F_j(x-),F_j(x)\right)$, then for $U\sim \unif$, $X_j=F_{j}^{-1}(U)\sim F_j$, and $\cov\left\{\cK_{j}\circ F_j^{-1}(U), G_j^{-1}(U)\right\}$
\begin{multline*}
   = \sum_{x\in \cA_j}\left\{  \cK_{j,F_j}(x) -\mu_j\right\}  \int_{F_j(x-)}^{F_j(x)} \left\{G_j^{-1}(u)-\tilde\mu_j\right\} du     + \int_{\{u\not \in \cI_j\}}\left\{K_j^{-1}(u)-\mu_j\right\}\left\{G_j^{-1}(u)-\tilde\mu_j\right\} du \\
   = \sum_{x\in\cA_j} \Delta_{F_j}(x) \left\{  \cK_{j,F_j}(x) -\mu_j\right\} \left\{  \cG_{j,F_j}(x) -\tilde\mu_j\right\}  + \int_{\{u\not \in \cI_j\}}\left\{K_j^{-1}(u)-\mu_j\right\}\left\{G_j^{-1}(u)-\tilde\mu_j\right\} du \\
  =\cov\left\{ \cK_{j}\circ F_j^{-1}(U),\cG_{j}\circ F_j^{-1}(U)\right\} =
  \cov\left\{ \cK_{j}(X_j),\cG_{j}(X_j)\right\}.
\end{multline*}
The rest of the proof follows from Remark \ref{rem:varsigma}.
\qed

\section{Supplementary material}\label{app:supp}

Inn this section, we verify the conditions for contiguity, which are 
 $C_\btheta$ has a continuous density $c_\btheta$  continuously differentiable  with square integrable gradient $\dot c_\btheta$ in a neighbourhood of $\btheta_0$, with $\dot c = \left. \nabla_\btheta c_\btheta (\bu)\right|_{\btheta=\btheta_0}$,  $\bu \in (0,1)^d$,  $\disp \dot C (\bu) = \int_{(0,\bu]} \dot c (\bs) d\bs$, and
\begin{equation}\label{eq:condvdVWnew}
\lim_{n\to \infty}\int_{(0,1)^d}  [ {n}^{1/2} [ \{  c_{\btheta_n} (\bu) \}^{1/2} -1] - \bdelta^\top \dot c
(\bu)/2  ]^2 d\bu = 0.
\end{equation}

Using the mean value theorem, one can see that the following stronger conditions implies \eqref{eq:condvdVWnew}:
\begin{eqnarray}\label{eq:condvdVW2anew}
\lim_{n\to \infty}\int_{(0,1)^d} \sup_{\|\btheta-\btheta_0\|\le n^{-1/2}\|\bdelta\|} \left\|\dot c_\btheta(\bu)- \dot c
(\bu)\right\|^2 d\bu &=& 0,\\
\limsup_{n\to \infty}\int_{(0,1)^d} \sup_{\|\btheta-\btheta_0\|\le n^{-1/2}\|\bdelta\|} \left\|\dot c_\btheta(\bu)\right\|^4 d\bu &<& \infty.
\label{eq:condvdVW2bnew}
\end{eqnarray}
The latter conditions are met for several bivariate copula families, including the Gaussian, Farlie-Gumbel-Morgenstern, Clayton and Frank. However they do not hold for Gumbel's copula since $\dot c$ is not square integrable.  \\

Note that $\dot c$ must be square integrable since if one uses LeCam's third lemma \citep{vanderVaart/Wellner:1996}, which is essential in proving contiguity results, one gets
 $\disp \sum_{i=1}^n \log{c_{\btheta_n}(\bU_i)} = \bdelta^\top \dZ_n - \frac{1}{2}\bdelta^\top \Sigma_n \bdelta+o+P(1)$ converges in law to $\bdelta^\top \dZ - \frac{1}{2}\bdelta^\top \Sigma \bdelta$,
 where $\dZ_n = n^{-1/2}\sum_{i=1}^n \dot c(\bU_i)$, $\Sigma_n = n^{-1}\sum_{i=1}^n \dot c(\bU_i) \dot c(\bU_i)^\top$, $\dZ\sim N(0,\Sigma)$, and $\disp \Sigma = \int_{(0,1)^d} \dot c(\bu) \dot c(\bu)^\top d\bu$.\\

\subsection{Clayton's copula}
In this case, one gets
\begin{eqnarray*}
\dot c_\theta(u,v) &=& -\log(u)-\log(v)+\frac{1}{1+\theta}+ \theta^{-2}\log\left(u^{-\theta}+v^{-\theta}-1\right)\\
&& \qquad +\left(\frac{1+2\theta}{\theta}\right) \frac{ \left\{ u^{-\theta} \log(u)+ v^{-\theta}\log(v)\right\}}{\left(u^{-\theta}+v^{-\theta}-1\right)}.
\end{eqnarray*}
As a result, $\dot c(u,v) = \left\{1+\log(u)\right\}\left\{1+\log(v)\right\}$, and $\dot c_\theta(u,v)-\dot c(u,v) = \theta g(u,v)+o(\theta)$, where
\begin{eqnarray*}
g(u,v) &=& \ln(u)^2+\ln(v)^2+8\ln(u)\ln(v)+\ln(v)^2\ln(u)^2+2\ln(v)+2\ln(u)\\
&& \qquad +3\ln(v)\ln(u)^2+3\ln(u)\ln(v)^2,
\end{eqnarray*}
and $\disp \int_{(0,1)^2}\{g(u,v)\}^4 dudv <\infty$.

\subsection{Farlie-Gumbel-Morgenstern}
For this copula family, $c_\theta(u,v)= 1+\theta(1-2u)(1-2v)$, so $\dot c_\theta(u,v)\equiv (1-2u)(1-2v)$, and \eqref{eq:condvdVW2anew}--\eqref{eq:condvdVW2bnew} are obviously met.
\subsection{Frank's copula}
In this case,  $\dot c(u,v) = \frac{(1-2u)(1-2v)}{2}$ and $\dot c_\theta(u,v)-\dot c(u,v)= \theta p_1(u,v)+o(\theta)$ uniformly over $(0,1)^2$, where $p_1(u,v)$ is a polynomial in $u$ and $v$. As a result, both \eqref{eq:condvdVW2anew}--\eqref{eq:condvdVW2bnew} are satisfied.

\subsection{Gaussian copula}

For the Gaussian copula with parameter $\rho$, the density is square integrable and a smooth function of $\rho$ about $\rho =0$, and conditions \eqref{eq:condvdVW2anew}--\eqref{eq:condvdVW2bnew} are met.

\subsection{Gumbel's copula}
This case is much more difficult. Here, setting $x=-\log(u)$ and $y=-\log(v)$, one gets
$$
\dot c(u,v) = \frac{1}{x+y}-(2-x-y)\log(x+y)-(y-1)\log(y)-(x-1)\log(x),
$$
which is not square integrable.  \\

\bibliographystyle{apalike}

\def\cprime{$'$}


\begin{table}[h!]
\caption{Power of  the proposed tests of serial independence for statistics $L_{n,2,5}$ and $L_{n,5}$ of Spearman's, van der Waerden's, and Savage's coefficients, for the independence copula and the Tent map copula, based on $N=1000$ replications.}\label{tab:ind+tent}
{\footnotesize
\begin{center}
\begin{tabular}{ccccccccccccccc}
\hline
& & \multicolumn{6}{c}{Ind} & \multicolumn{6}{c}{Tent map}\\
$n$  & Margin &  \multicolumn{2}{c}{Spearman} & \multicolumn{2}{c}{van der Waerden} & \multicolumn{2}{c}{Savage} &  \multicolumn{2}{c}{Spearman} & \multicolumn{2}{c}{van der Waerden} & \multicolumn{2}{c}{Savage}\\
\cline{3-14}\\
& & $L_{n,2,5}$ & $L_{n,5,5}$ & $L_{n,2,5}$ &  $L_{n,5,5}$ & $L_{n,2,5}$ & $L_{n,5,5}$ & $L_{n,2,5}$ & $L_{n,5,5}$ & $L_{n,2,5}$ &  $L_{n,5,5}$ & $L_{n,2,5}$ & $L_{n,5,5}$\\
\\
\multirow{7}{*}{100}  & $F_1$     &   3.6 &  6.4 &  3.6   & 6.4  &  3.6   & 6.4  &    84.8 &  77.3 &  84.8  & 77.3 &  84.8 & 77.3 \\
                      & $F_2$     &   5.1 &  5.6 &  4.5   & 7.5  &  3.9   & 6.1  &     9.1 &  17.4 &   6.8  & 45.0 &  61.1 &  62.7  \\
                      & $F_3$     &   6.0 &  5.0 &  5.3   & 7.0  &  5.5   & 7.9  &     9.2 &  15.0 &   5.8  & 43.1 &  68.5 &  67.5  \\
                      & $F_4$     &   5.8 &  4.8 &  5.4   & 6.3  &  4.6   & 7.5  &     7.7 &  16.3 &   4.3  & 50.8 &  68.0 &  65.4  \\
                      & $F_5$     &   5.2 &  5.9 &  5.0   & 6.3  &  3.3   & 5.3  &     7.9 &  16.6 &   5.4  & 56.2 &  59.6 &  56.2  \\
                      & $F_6$     &   4.6 &  5.1 &   4.7  & 7.1  &  3.4   & 5.5  &     8.5 &  16.3 &   6.2  & 54.3 &  59.7 &  61.4  \\
                      & $F_7$     &   5.5 &  4.3 &   5.7  & 5.4  &  5.7   & 4.7  &    52.0 &  32.3 &  57.1  & 34.8 &  16.8 &  12.2   \\
                     \hline
 \multirow{7}{*}{250} & $F_1$     &   5.5 &  6.0 &   5.5  & 6.0  &  5.5   & 6.0  &   100.0 &  99.7 & 100.0  & 99.7 & 100.0 &  99.7  \\
                      & $F_2$     &   4.6 &  4.8 &   5.2  & 5.7  &  4.4   & 6.6  &    10.1 &  28.2 &   8.2  & 68.9 &  97.6 &  98.4   \\
                      & $F_3$     &   4.3 &  5.2 &   4.3  & 5.9  &  6.4   & 7.9  &     9.9 &  22.8 &   9.6  & 60.9 &  97.9 &  97.4   \\
                      & $F_4$     &   4.6 &  5.2 &   4.4  & 6.4  &  4.4   & 6.5  &    10.7 &  29.1 &   7.7  & 64.4 &  97.1 &  97.2   \\
                      & $F_5$     &   6.1 &  5.3 &   5.8  & 6.6  &  4.1   & 9.8  &     8.1 &  23.7 &   6.1  & 72.0 &  98.1 &  98.2   \\
                      & $F_6$     &   3.5 &  4.4 &   3.9  & 5.3  &  4.1   & 8.0  &     7.9 &  24.3 &   7.2  & 71.9 &  98.3 &  98.2   \\
                      & $F_7$     &   4.2 &  5.3 &   3.6  & 5.3  &  5.0   & 5.2  &    90.9 &  78.7 &  98.1  & 92.5 &  36.2 &  24.7   \\
                      \hline
 \multirow{7}{*}{500} & $F_1$     &   5.4 &  7.3 &   5.4  & 7.3  &  5.4   & 7.3  &   100.0 & 100.0 & 100.0  &100.0 & 100.0 & 100.0   \\
                      & $F_2$     &   4.8 &  4.2 &   4.8  & 5.3  &  4.9   & 7.7  &    11.7 &  33.4 &  11.2  & 75.3 & 100.0 & 100.0  \\
                      & $F_3$     &   5.0 &  5.0 &   5.0  & 5.3  &  5.5   & 6.8  &    12.0 &  29.8 &  18.9  & 73.3 & 100.0 & 100.0  \\
                      & $F_4$     &   3.9 &  5.8 &   3.6  & 5.4  &  5.4   & 5.6  &    12.0 &  37.8 &  12.3  & 80.4 & 100.0 &  99.9  \\
                      & $F_5$     &   5.4 &  5.1 &   5.9  & 5.9  &  4.7   & 7.6  &    11.8 &  27.2 &   9.4  & 79.5 & 100.0 & 100.0  \\
                      & $F_6$     &   4.1 &  5.1 &   3.8  & 4.7  &  3.7   & 7.1  &     9.2 &  29.0 &   8.8  & 80.5 & 100.0 & 100.0  \\
                      & $F_7$     &   4.9 &  4.9 &   5.1  & 5.5  &  5.0   & 5.4  &    99.8 &  99.2 & 100.0  &100.0 &  65.3 &  49.5 \\
                     \hline
\end{tabular}
\end{center}
}
\end{table}

\begin{table}[h!]
\caption{Power of  the proposed tests of serial independence for statistics $L_{n,2,5}$ and $L_{n,5,5}$ of Spearman's, van der Waerden's, and Savage's coefficients, for the Farlie-Gumbel-Morgenstern  and Clayton copula families, based on $N=1000$ replications.}\label{tab:fgm+clayton}
{\footnotesize
\begin{center}
\begin{tabular}{ccccccccccccccc}
\hline
& & \multicolumn{6}{c}{FGM} & \multicolumn{6}{c}{Clayton}\\
$n$  & Margin &  \multicolumn{2}{c}{Spearman} & \multicolumn{2}{c}{van der Waerden} & \multicolumn{2}{c}{Savage} &  \multicolumn{2}{c}{Spearman} & \multicolumn{2}{c}{van der Waerden} & \multicolumn{2}{c}{Savage}\\
\cline{3-14}\\
& & $L_{n,2,5}$ & $L_{n,5,5}$ & $L_{n,2,5}$ &  $L_{n,5,5}$ & $L_{n,2,5}$ & $L_{n,5,5}$ & $L_{n,2,5}$ & $L_{n,5,5}$ & $L_{n,2,5}$ &  $L_{n,5,5}$ & $L_{n,2,5}$ & $L_{n,5,5}$\\
\\
\multirow{7}{*}{100}  & $F_1$     &   4.4  & 10.7 &  4.4  & 10.7 &   4.4  & 10.7   &  15.3 &  24.6 & 15.3  & 24.6  & 15.3 &  24.6 \\
                      & $F_2$     &   5.9  & 15.6 &  5.7  & 13.4 &   7.0  &  8.7   &  17.1 &  11.8 & 18.7  & 17.2  & 27.6 &  30.8   \\
                      & $F_3$     &   4.7  & 14.7 &  5.1  & 14.2 &   6.5  &  9.8   &  18.5 &  13.8 & 17.0  & 15.3  & 23.0 &  28.6   \\
                      & $F_4$     &   5.8  & 14.1 &  5.3  & 14.1 &   6.8  & 10.7   &  17.1 &  12.4 & 16.8  & 15.2  & 25.9 &  29.4   \\
                      & $F_5$     &   5.6  & 14.6 &  4.9  & 12.4 &   5.3  &  8.8   &  18.0 &  12.4 & 18.7  & 16.6  & 25.0 &  26.7   \\
                      & $F_6$     &   5.4  & 14.0 &  3.7  & 10.0 &   4.3  &  7.8   &  16.8 &  12.2 & 18.0  & 16.2  & 29.6 &  26.8   \\
                      & $F_7$     &   5.9  & 12.0 &  5.7  & 10.1 &   6.2  & 12.4   &   9.8 &   8.9 &  9.0  &  8.6  & 10.2 &   8.9    \\
                     \hline
 \multirow{7}{*}{250} & $F_1$     &   7.2  & 15.7 &  7.2  & 15.7 &   7.2  & 15.7   &  36.4 &  40.9 & 36.4  & 40.9  & 36.4 &  40.9   \\
                      & $F_2$     &   6.8  & 38.5 &  6.6  & 32.4 &   8.1  & 16.2   &  39.7 &  26.3 & 43.5  & 35.5  & 61.9 &  53.3    \\
                      & $F_3$     &   7.5  & 38.2 &  7.7  & 32.1 &   8.7  & 20.9   &  39.8 &  29.4 & 40.8  & 31.8  & 59.8 &  53.8    \\
                      & $F_4$     &   8.0  & 39.8 &  8.3  & 32.9 &   8.4  & 22.3   &  40.1 &  26.0 & 41.5  & 28.8  & 57.1 &  50.5    \\
                      & $F_5$     &   7.8  & 39.9 &  8.4  & 31.5 &   7.9  & 14.5   &  46.3 &  29.0 & 51.1  & 38.0  & 68.0 &  61.5    \\
                      & $F_6$     &   6.0  & 40.7 &  6.3  & 31.2 &   7.4  & 16.4   &  45.9 &  31.9 & 51.0  & 41.3  & 66.5 &  58.4    \\
                      & $F_7$     &   6.9  & 25.5 &  5.6  & 24.1 &   6.2  & 24.1   &  18.1 &  11.3 & 16.8  & 11.6  & 19.2 &  12.2    \\
                      \hline
 \multirow{7}{*}{500} & $F_1$     &   6.0  & 19.1 &  6.0  & 19.1 &   6.0  & 19.1   &  60.9 &  61.7 & 60.9  & 61.7  & 60.9 &  61.7    \\
                      & $F_2$     &   7.6  & 74.8 &  6.6  & 64.4 &   8.6  & 28.0   &  76.3 &  58.5 & 82.1  & 70.4  & 92.0 &  85.8   \\
                      & $F_3$     &   8.0  & 77.7 &  7.4  & 70.1 &   8.0  & 42.9   &  73.3 &  55.5 & 76.2  & 59.4  & 87.4 &  81.2   \\
                      & $F_4$     &  10.0  & 76.3 &  8.8  & 69.5 &  10.6  & 42.2   &  73.9 &  55.3 & 77.4  & 58.2  & 87.3 &  78.3   \\
                      & $F_5$     &   9.1  & 78.6 &  9.1  & 67.6 &   8.1  & 28.6   &  75.2 &  54.1 & 81.7  & 64.8  & 92.5 &  82.9   \\
                      & $F_6$     &   9.2  & 77.5 &  9.7  & 67.7 &   9.2  & 29.2   &  74.3 &  53.9 & 80.3  & 66.8  & 91.1 &  83.3   \\
                      & $F_7$     &   8.4  & 58.0 &  6.6  & 52.0 &   8.2  & 51.2   &  39.3 &  23.4 & 33.4  & 22.0  & 41.2 &  23.5  \\
                     \hline
\end{tabular}
\end{center}
}
\end{table}

\begin{table}[h!]
\caption{Power of  the proposed tests of serial independence for statistics $L_{n,2,5}$ and $L_{n,5,5}$ of Spearman's, van der Waerden's, and Savage's coefficients, for the Gaussian  and Frank copula families, based on $N=1000$ replications.}\label{tab:gaussian+frank}
{\footnotesize
\begin{center}
\begin{tabular}{ccccccccccccccc}
\hline
& & \multicolumn{6}{c}{Gaussian} & \multicolumn{6}{c}{Frank}\\
$n$  & Margin &  \multicolumn{2}{c}{Spearman} & \multicolumn{2}{c}{van der Waerden} & \multicolumn{2}{c}{Savage} &  \multicolumn{2}{c}{Spearman} & \multicolumn{2}{c}{van der Waerden} & \multicolumn{2}{c}{Savage}\\
\cline{3-14}\\
& & $L_{n,2,5}$ & $L_{n,5,5}$ & $L_{n,2,5}$ &  $L_{n,5,5}$ & $L_{n,2,5}$ & $L_{n,5,5}$ & $L_{n,2,5}$ & $L_{n,5,5}$ & $L_{n,2,5}$ &  $L_{n,5,5}$ & $L_{n,2,5}$ & $L_{n,5,5}$\\
\\
\multirow{7}{*}{100}  & $F_1$     &    8.7 &  14.4  &  8.7 &  14.4  &  8.7 &  14.4   &     8.4 &  16.2  &  8.4 & 16.2 &   8.4  & 16.2  \\
                      & $F_2$     &   16.2 &  10.4  & 16.4 &   9.9  & 11.4 &  12.9   &    17.3 &  10.6  & 15.9 &  8.6 &  11.5  & 11.6    \\
                      & $F_3$     &   16.0 &  10.5  & 15.8 &  13.2  & 11.8 &  13.0   &    15.2 &  10.0  & 13.9 &  9.9 &  10.9  & 12.4    \\
                      & $F_4$     &   14.8 &   9.6  & 14.5 &  11.2  & 13.8 &  14.5   &    16.6 &   8.8  & 14.4 &  9.6 &  11.4  & 12.1    \\
                      & $F_5$     &   15.3 &  11.5  & 15.8 &  11.6  & 12.8 &  11.7   &    16.3 &  10.9  & 14.4 & 11.6 &  12.1  & 12.2    \\
                      & $F_6$     &   13.1 &   9.3  & 14.1 &  11.3  & 10.5 &  12.3   &    15.8 &   8.7  & 13.9 & 10.2 &   9.4  &  9.4    \\
                      & $F_7$     &   12.4 &   7.9  & 13.1 &  12.1  & 11.6 &   6.4   &    13.1 &   9.7  & 13.1 &  9.8 &  13.4  &  8.9     \\
                     \hline
 \multirow{7}{*}{250} & $F_1$     &   14.5 &  19.4  & 14.5 &  19.4  & 14.5 &  19.4   &    14.2 &  21.5  & 14.2 & 21.5 &  14.2  & 21.5    \\
                      & $F_2$     &   40.2 &  21.7  & 41.3 &  23.6  & 31.2 &  24.3   &    38.4 &  22.6  & 33.0 & 22.6 &  24.5  & 18.5     \\
                      & $F_3$     &   41.5 &  23.8  & 43.3 &  25.9  & 31.1 &  25.6   &    40.9 &  23.2  & 38.4 & 22.5 &  28.0  & 20.8     \\
                      & $F_4$     &   40.6 &  22.5  & 42.7 &  24.5  & 32.7 &  25.1   &    43.6 &  24.6  & 39.7 & 23.1 &  29.2  & 22.2     \\
                      & $F_5$     &   42.4 &  26.5  & 44.2 &  27.9  & 31.5 &  27.1   &    40.8 &  22.4  & 36.8 & 21.4 &  22.7  & 17.2     \\
                      & $F_6$     &   38.4 &  22.9  & 43.4 &  26.6  & 32.3 &  25.1   &    42.4 &  24.4  & 38.5 & 20.9 &  24.1  & 16.0     \\
                      & $F_7$     &   31.4 &  20.9  & 31.8 &  25.2  & 27.3 &  17.3   &    32.4 &  19.8  & 32.3 & 19.8 &  29.5  & 18.1     \\
                      \hline
 \multirow{7}{*}{500} & $F_1$     &   27.2 &  28.8  & 27.2 &  28.8  & 27.2 &  28.8   &    23.8 &  26.5  & 23.8 & 26.5 &  23.8  & 26.5     \\
                      & $F_2$     &   73.3 &  49.7  & 78.1 &  56.2  & 60.1 &  42.9   &    76.2 &  48.0  & 71.0 & 41.9 &  47.5  & 31.2    \\
                      & $F_3$     &   76.0 &  50.5  & 77.8 &  55.4  & 61.1 &  42.6   &    74.6 &  50.1  & 72.4 & 45.6 &  54.5  & 34.3    \\
                      & $F_4$     &   73.8 &  48.6  & 77.0 &  52.6  & 63.3 &  44.2   &    75.2 &  52.6  & 71.7 & 48.7 &  54.6  & 38.3    \\
                      & $F_5$     &   75.9 &  51.3  & 79.5 &  56.5  & 60.9 &  39.5   &    73.4 &  49.7  & 69.3 & 43.5 &  47.6  & 27.3    \\
                      & $F_6$     &   75.5 &  52.9  & 80.0 &  55.1  & 63.2 &  43.0   &    74.3 &  46.7  & 68.3 & 43.4 &  46.3  & 29.1    \\
                      & $F_7$     &   56.8 &  36.5  & 59.5 &  42.5  & 48.7 &  30.6   &    61.3 &  41.3  & 58.8 & 39.4 &  58.6  & 36.5   \\
                     \hline
\end{tabular}
\end{center}
}
\end{table}

\end{document}